\documentclass[12pt]{article}
\usepackage{amsmath,latexsym,amssymb,graphicx,epsf,amsfonts,amsthm}
\usepackage{epsf}
\usepackage{color}
\numberwithin{equation}{section} \numberwithin{figure}{section}

\newtheorem{thm}{Theorem}
\newtheorem{theorem}{Theorem}

\newtheorem*{ass}{Assumption}
\newtheorem{defi}[thm]{Definition}
\newtheorem{rem}[thm]{Remark}
\theoremstyle{definition}
\newtheorem{exmp}{Example}[section]
\begin{document}
\hoffset = -1truecm \voffset = -1truecm
\title{Lorentzian Peak Sharpening and Sparse Blind Source Separation for NMR Spectroscopy}
\author{Yuanchang Sun\thanks{Department of Mathematics and Statistics, Florida International University, Miami FL 33189, USA.}, Jack Xin\thanks{Department of Mathematics,
University of California at Irvine, Irvine, CA 92697, USA.} }
\date{}
\maketitle

\begin{abstract}
In this paper, we introduce a preprocessing technique for blind source separation (BSS) of nonnegative and overlapped data.  For Nuclear Magnetic Resonance spectroscopy (NMR),  the classical method of Naanaa and Nuzillard (NN)  requires the condition that source signals to be non-overlapping at certain locations while they are allowed to overlap with each other elsewhere.  NN's method works well with data signals that possess stand alone peaks (SAP). The SAP does not hold completely for realistic NMR spectra however.  Violation of SAP often introduces errors or artifacts in the NN's separation results.  To address this issue, a preprocessing technique is developed here based on Lorentzian peak shapes and weighted peak sharpening.  The idea is to superimpose the original peak signal with its weighted negative second order derivative.  The resulting sharpened (narrower and taller) peaks enable NN's method to work with a more relaxed SAP condition, the so called dominant peaks condition (DPS), and deliver improved results.  To achieve an optimal sharpening while preserving the data nonnegativity,  we prove the existence of an upper bound of the weight parameter and propose a selection criterion.   Numerical experiments on NMR spectroscopy data show satisfactory performance of our proposed method.

\end{abstract}
%
%
%
%
%
\setcounter{equation}{0} \setcounter{page}{1}
\section{Introduction}
In applications such as computer tomography, biomedical imaging, and spectroscopic sensing, the data collected are usually nonnegative and correlated, and the objects being imaged are often mixtures of substances, which pose a serious challenge for direct identification and quantification of the constituents. In many situations, we need to decompose the data into a set of basic components (source signals) without knowing the mixing process, or solve a blind source separation (BSS) problem.

The objective of BSS is to extract a number of source signals from their linear mixtures without the knowledge of the mixing process.  BSS has been playing a central role in a wide range of signal and image processing problems such as speech recognition, sound unmixing, image separations, and text mining, to name a few \cite{choi,Cic, Comon1}.  In this paper we are interested in a BSS problem arising from the Nuclear Magnetic Resonance (NMR) spectroscopy \cite{nmr}.  Being one of the preeminent imaging techniques in chemistry, NMR spectroscopy is frequently used by chemists and biochemists to study the molecular structures of organic compounds.  NMR spectroscopy and other imaging techniques have made it possible to identify and classify pure substances by their fingerprint spectra.  The real world data however may involve multiple unknown substances besides impurities, and are subject to background and environment changes. This makes the data analysis hopeless unless we can unmix or separate the mixed data into a list of source components.  In many practical situations, we need to determine from a mixture the constituent chemicals and their coefficients as a BSS problem whose mathematical model takes the following matrix form;
\begin{equation}
\label{BSS}
 X= A\, S + N\;
\end{equation}
where $X\in \mathbb{R}^{m\times p}, A \in \mathbb{R}^{m\times b}, S \in \mathbb{R}^{n\times p} $.
Rows of $X$ represents the spectral mixtures, rows of $S$ are the source signals, and entries of matrix $A$ are the mixing coefficients, $N$ is the noise matrix.  The goal of BSS is to solve for $A$ and $S$ given $X$. If $P$ is a permutation matrix and $D$ an invertible diagonal matrix, one can immediately notice that $ AS = (APD)(D^{-1}P^{- 1}S)$,  hence $(A,S)$ and $(APD,D^{-1}P^{-1}S)$ are considered equivalent solutions in BSS.

\medskip

There have been mainly two classes of BSS methods for solving (\ref{BSS}). The first class of methods belong to statistical regime.  Among others, independent component analysis (ICA) is the most well studied statistical BSS approach, it decomposes a mixed signal into additive source components based on the mutual independence of the non-Gaussian source signals.  The statistical independence requires uncorrelated source signals, and this condition however is not always satisfied by realistic data.  For example, the statistical independence does not hold in the NMR spectra of chemical compounds where molecules responsible for each source share common structural features.
The deterministic BSS methods include nonnegative matrix factorization (NMF) and geometrical methods. Introduced by Paatero and Tapper \cite{NMF0} and popularized by Lee and Seung \cite{Lee},  NMF has become the prevalent method for solving nonnegative BSS problems.  NMF seeks a factorization of $X$ into product of two nonnegative matrices by minimizing the cost function of a certain distance or divergence metric \cite{choi}.  NMF does not impose source independence, however, some additional constraints such as sparsity of the sources and/or the mixing matrix, are often imposed to control the non-uniqueness.  In \cite{NMF_OR}, the orthogonality (correlation) constraints and prior knowledge of a target spectrum are incorporated into NMF to guide the factorization and improve the effectiveness in chemical agent detection.  In \cite{Miao,Sch}, NMF is augmented with a minimum determinant constraint on the estimated mixing matrix to tackle the non-uniqueness.   Although they have been successful in some BSS problems, the NMF and ICA are both non-convex methods which can be unreliable in decomposing real world data.  Geometrical BSS methods are based on convex geometry of the data matrix $X$. The columns of $X$ are nonnegative linear combinations of those of $A$. In the hyperspectral unmixing (HSI) setting, a condition called pure pixel assumption (PPA) was proposed in \cite{Chang_07} which requires the presence in the data of at least one pure pixel of each endmember (source signal). In NMR spectroscopy, PPA was reformulated by Naanaa and Nuzillard \cite{NN05}.  The source signals are only required to be non-overlapping at some locations of acquisition variable.  This condition was applied to NMR data unmixing and led to a major success of a convex cone method.  Such a local sparseness condition greatly reduces this problem to a convex one which is solvable by linear programming.  Though the convex cone method is geometrically elegant, the working condition is still restrictive. In fact, NNA or PPA is not always satisfied in either NMR or HSI. Within the convex framework, a recent work of the authors studied how to postprocess with the abundance of mixture data, and how to improve mixing matrix estimation with major peak based corrections when the strict sparseness in NNA is violated mildly \cite{sun_xin_pNN}.   Other geometrical methods include minimum volume cone method which is to fit a simplex (convex cone) of minimum volume to the data set \cite{MVT, MVSA}. This method is a non-convex approach which amounts to solving a minimization problem by finding a matrix with minimum volume under a constraint.

In the present work, we are concerned with a class of NMR spectral data from chemicals sharing common molecular structures.  Hence their spectra should consist of similar peak components.  In fact, the sparseness condition NNA proposed by Naanaa and Nuzillard can be interpreted as a stand alone peak condition (SAP) for NMR data with peak components.  That is, each source signal possesses a stand alone peak extending over an acquisition interval while other source spectra are identically zero over this interval.  In this paper, we consider how to generalize NN method if the SAP condition is not satisfied strictly.  We shall consider a regime where the source signals have dominant peaks (DPS) over one another on certain acquisition intervals.  The idea is to sharpen these peaks (shrink the dominant intervals) so that the dominant peaks approximately become stand alone peaks, hence to improve the NN separation results.  In the context of image enhancement (for example deblurring),  Kovasznay and Joseph \cite{Kov} in 1955 found that a blurred image could be deblurred and sharpened by subtracting a fraction of its Laplacian
\begin{equation*}
U_{\mathrm{e}} = U_{\mathrm{o}} -k\Delta U_{\mathrm{o}}\;,
\end{equation*}
where $U_{\mathrm{o}}$ represents the original image, $U_{\mathrm{e}}$ the enhanced image.   This idea can be applied to signals to sharpen their peaks and enhance the resolution.  Note that a NMR spectrum can be expressed as the nonnegative linear combinations of Lorentz functions, as shown in Fig. \ref{lorentz example}.   To sharpen the Lorentzian peaks, we subtract a weighted second order derivative from the original signal to enhance the resolution.
\begin{equation*}
\hat{S} = S- k S''\;,
\end{equation*}
where $\hat{S}$ is the sharpened signal, $S$ the original signal, $S''$ is the second order derivative, and $k$ is the weight parameter whose selection will be discussed in detail later.  The sharpening makes the peaks narrower with enhanced resolution so they approximately become stand alone peaks.  After the preprocessing is accomplished, the NN approach is then applied to retrieve the mixing matrix $A$.  The separation of the source signals may be solved by a nonnegative least squares method.

The paper is organized as following: In section 2, we shall briefly review the NN method and its partial sparseness condition, then state the more suitable stand alone peaks and dominant peaks assumptions for NMR data.  In section 3, we present the weighted peak sharpening method and its mathematical analysis.  A selection criterion of the weight parameter is proposed for optimal sharpening and data nonnegativity.  In section 4, numerical experiments are performed to test the effectiveness of the proposed method.  Concluding remarks are in section 5.

\begin{figure}
\includegraphics[height=5cm,width=13cm]{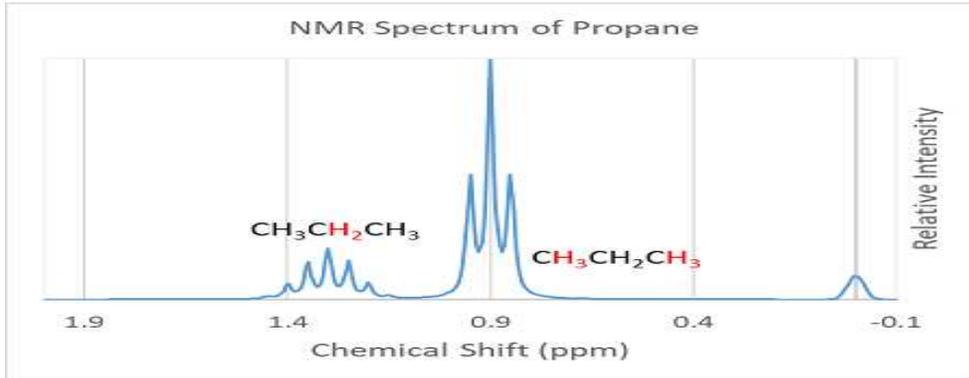}
\caption{NMR spectrum of an organic compound : Propane. hydrogens would be split into two peaks (a doublet), and the aldehyde H into four peaks (a quartet). Source: {\sl www.study.com}.}
 \label{lorentz example}
\end{figure}
\section{Sparse BSS and Geometric Constructions}
\subsection{NN's Method}
In this part, we shall review NN's method for nonnegative and overlapped data \cite{NN05}.  The working criterion of their method is a local sparseness assumption on source signals;  it is that the signals are only allowed to be non-overlapping at certain acquisition locations, while they might overlap with each other elsewhere.  Mathematically speaking, the source matrix $S$ needs to satisfy the following assumption (recall that $m$ is the number of mixed signals, $n$ the number of source signals, and $p$ the number of samples):
\begin{ass}[NNA]
For each $i\in\{1,2,\dots,n \}$ there is an $j_{i}\in
\{1,2,\dots,p\}$ such that $s_{i,j_i}>0$ and $s_{k,j_i}=0\;
(k=1,\dots,i-1,i+1,\dots,n)\;.$
\end{ass}
Let us consider equation (\ref{BSS}) in terms of columns
\begin{equation}
\label{LinComb}
 X^{j} = \sum^{n}_{k=1}s_{k,j}A^{k}, \;\;\;\;\;\; j = 1,\dots,p,
\end{equation}
then $\displaystyle X^{j_i} = s_{i,j_i}A^i$, $i =
1,\dots,n\;\; $ or $A^{i} = \frac{1}{s_{i,j_i}}X^{j_i}$ by the NNA condition.  Therfore
equation
(\ref{LinComb}) can be expressed as
\begin{equation}
\label{NNlinComb} X^{j} = \sum^{n}_{i = 1}
\frac{s_{i,j}}{s_{i,j_i}}X^{j_i}\;,
\end{equation}
which implies that every column of $X$ is in fact a nonnegative linear
combination of the columns of the matrix $[X^{j_1},\dots,X^{j_n}]$.  Denote $\hat{A} =
[X^{j_1},\dots,X^{j_n}]$, a submatrix of $X$ with $n$
columns.  Examining equations (\ref{LinComb}) and (\ref{NNlinComb}), we see that
each column of $\hat{A}$ is collinear to a particular column of $A$ .
Once all the $j_i$'s are found, an estimation of
the mixing matrix is achieved.  The identification of $\hat{A}$'s
columns is equivalent to identifying the edges of a convex cone that encloses the data
columns of $X$.  For a noiseless case $X = AS$, the following constrained equations
are formulated for the identification of $\hat{A}$,
\begin{equation}
\label{LPNF} \sum^{p}_{j = 1, j\neq k}X^{j} \lambda_j = X^{k},\;\;\;\;
\lambda_j\geq 0,\;\;\;\; k = 1,\dots,p.
\end{equation}
Then a column vector $X^{k}$ will be a column of $\hat{A}$ if and
only if the constrained equation (\ref{LPNF}) is inconsistent (has
no solution $X^j$, $j\not =k$).  The Moore-Penrose inverse
$\hat{A}^{+}$ of $\hat{A}$ is then calculated and an estimate of $S$
is obtained: $\hat{S} = \hat{A}^{+} X$.

As it applies to NMR spectra with peak, NNA can be restated as the stand alone peak (SAP) condition: each source signal possesses a stand alone peak over certain acquisition interval, where other sources are identically zero.  Precisely the source matrix $S$ should satisfy the following condition:
\begin{ass}[SAP] For each
$i\in\{1,2,\dots,n \}$ there exists a set of consecutive integers $\mathcal{I}\subset \{1,2,\dots,p\} $ such that $S_{i,k} >0 $ for $k\in \mathcal{I}$ and $ S_{j, k} = 0\; (j=1,\dots,i-1,i+1,\dots,n)\;.$
\end{ass}
\begin{figure}
\includegraphics[height=5cm,width=8cm]{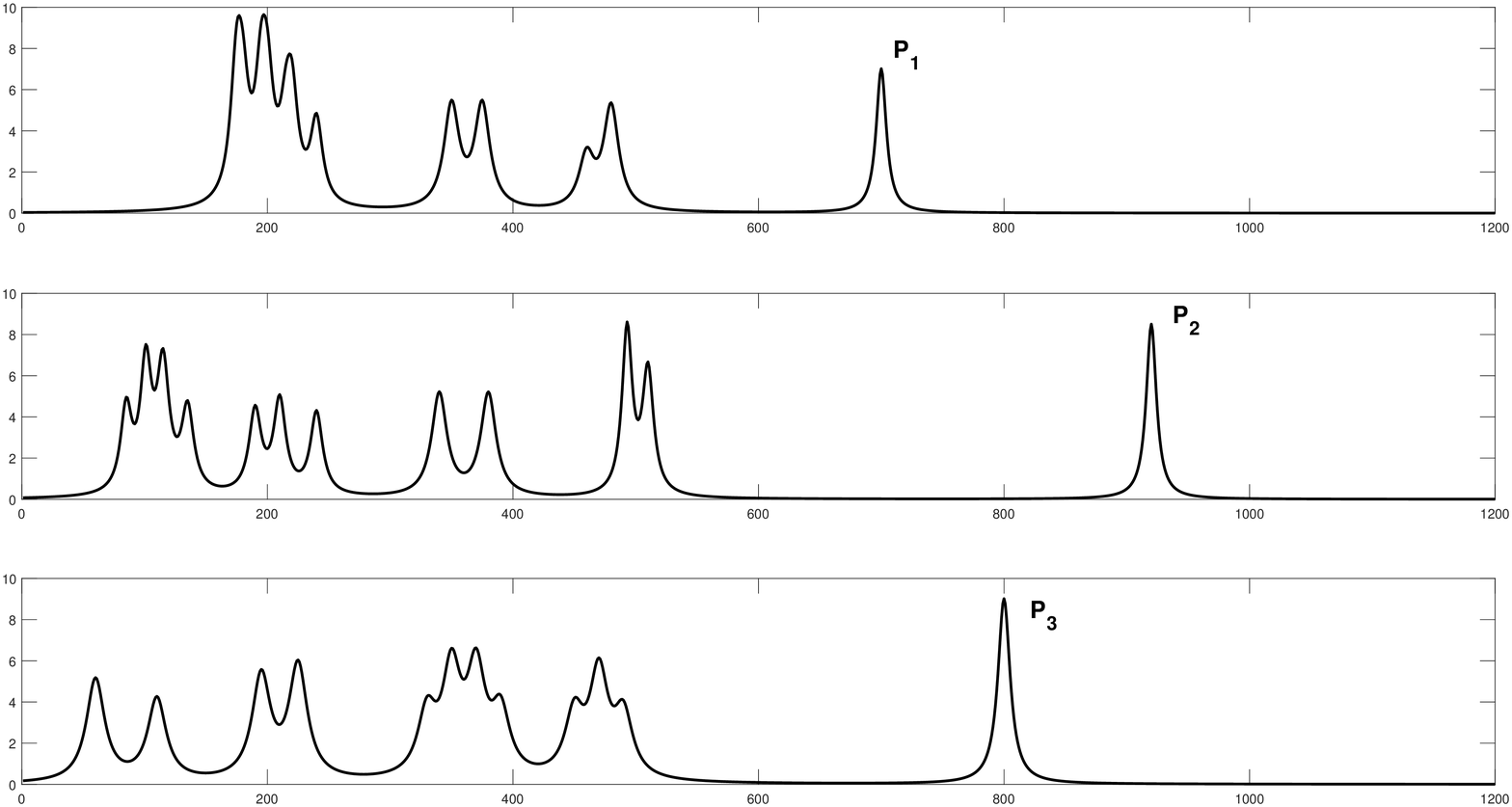}
\includegraphics[height=5cm,width=8cm]{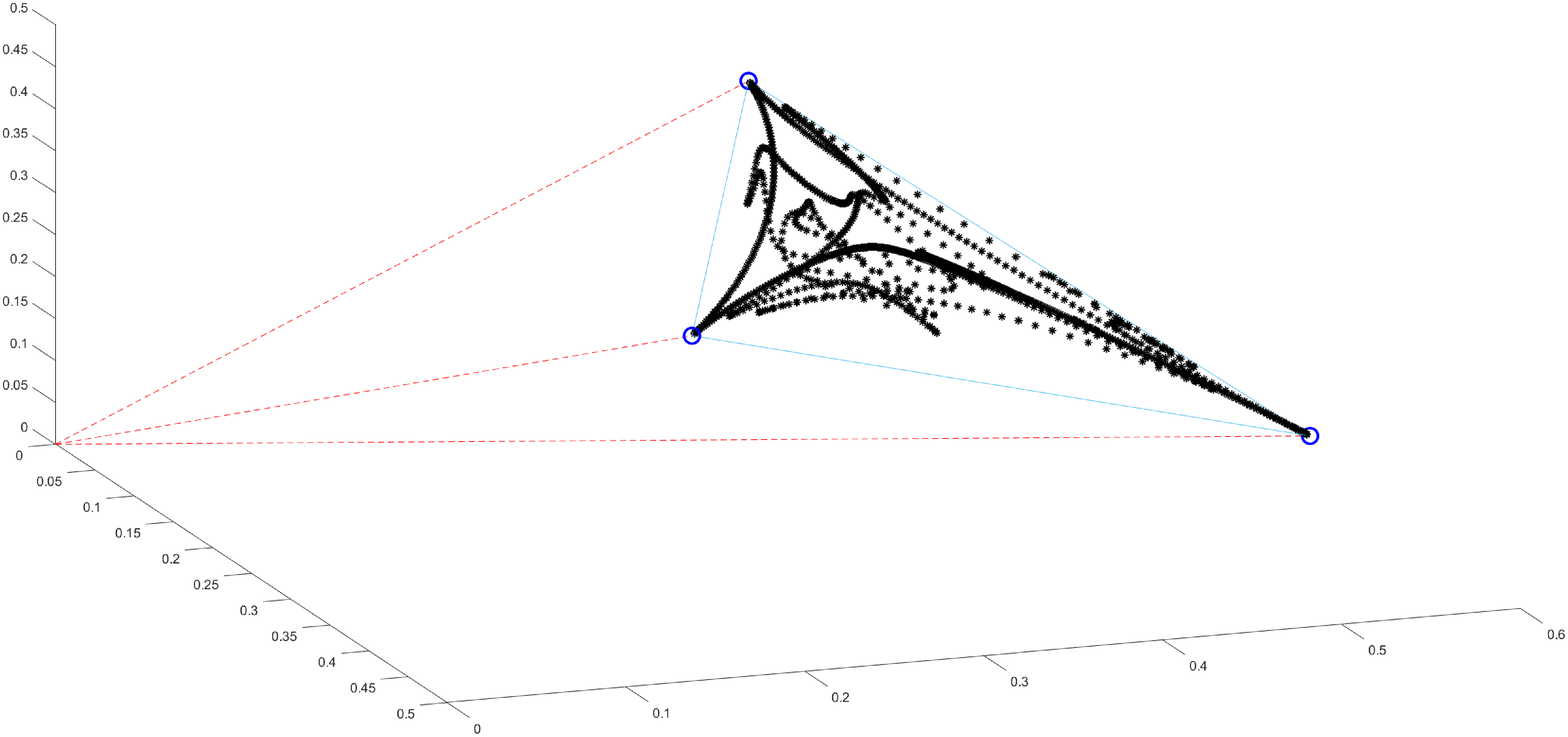}
\caption{Left: the synthetic Lorentzian NMR spectra of three SAP sources.  Each spectrum has a stand alone peak indicated by $P_1,P_2$, and $P_3$.  Right: the scattered plots of $X$ (columns of $X$) scaling to be on plane $x+y+z = 1 $.}
 \label{NNA}
\end{figure}
The SAP condition is illustrated by NMR spectra of three sources in the left plot of Fig. \ref{NNA}, it can be seen that each source signal has a stand alone peak denoted by $P_1,P_2,$ and $P_3$, respectively.  In this illustrative example, there are three mixtures and three sources for the linear mixture model (\ref{BSS}):
\begin{exmp}
\label{example1}
 $X_{3\times p} = A_{3\times 3} S_{3\times p}$, we shall view each column of $X$ as a point in the 3-space, then
{
\begin{equation*}
\noindent \left [ X^1,X^2,\cdots,X^p \right ]
\end{equation*}
\begin{equation*}
 = \left [ A^1,A^2,A^3 \right ]\times
\left(
\begin{array}{cccccccccc}
        * & \cdots & * & {\bf \color{blue}\mathbf{u} }     & {\color{blue}o}  & {\color{blue}o}  & * & \cdots & *\\
        * & \cdots & * & {\color{blue}o }     & {\bf \color{red}\mathbf{v}}  & {\color{blue}o}  & * & \cdots & *\\
        * & \cdots & * & {\color{blue}o }     & {\color{blue}o}  &  {\bf \color{black}\mathbf{w}} & * & \cdots & *
       \end{array}
 \right)\;,
\end{equation*}}
Here $\mathbf{u,v,w}$ are the stand alone peaks from the three source signals.
\end{exmp}
These stand alone peaks span a convex cone enclosing all the columns of $X$, illustrated in the right plot in Fig. \ref{NNA}.  The estimation of $A$ is then equivalent to the identification of this cone.  To do so, the following optimization problem is solved for each scaled column of $X$ (i.e., the columns are scaled to be on a plane)
\begin{equation}
\label{NNCG}
c = \min \sum^p_{j = 1, j\neq k} \lambda_j ,\;\; \mathrm{such\; that} \; \sum^{p}_{j = 1, j\neq k}X(:,j) \lambda_j = X(:,k)\;\;, \lambda_j\geq 0.
\end{equation}
It is shown in \cite{Dul} that $X(:,k)$ is a vertex of the convex
cone if and only if the optimal objective function value $c^*$ is greater than 1.   Once $A$ is located, $S$ maybe thereafter recovered by nonnegative least squares.  This geometric construction of $A$ is also called vertex component analysis (VCA).

\subsection{Relaxation of Stand Alone Peaks: Dominant Peaks}
The NN method proves to be successful in separating data signals if the working condition is strictly satisfied. The real-world data may not satisfy the SAP completely due to measurement noises or the underlying physical process,  consequently the NN's method might introduce errors (spurious peaks) in the output. It is more realistic to assume that the signals are positive extending over the whole acquisition range and stand alone peaks could overlap to some extent, that is the stand alone peaks become dominant peaks (DPS).  More formally, the source matrix is required to satisfy the following condition.
\begin{ass}[DPS] For each
$i\in\{1,2,\dots,n \}$ there exists a set of consecutive integers $\mathcal{I}\subset \{1,2,\dots,p\} $ such that $S_{i, k} >0 $ for $k\in \mathcal{I}$ and $ S_{j, k} = \epsilon_{k_j} \ll S_{i,k} \; (j=1,\dots,i-1,i+1,\dots,n)\;.$
 \end{ass}
Simply said, each source signal has a dominant peak over an acquisition interval where the other sources are allowed to be nonzero.  DPS condition is more appropriate for NMR spectra consisting of positive-valued peaks with tails extending over the whole range of acquisition variable. In DPS signals, the previous example (\ref{example1}) of three source signals matrix $S$ would look like
\begin{equation*}
S=
\left(
\begin{array}{cccccccccc}
        * & \cdots & * & {\bf \color{blue}\mathbf{u} }     & {\color{blue}\mathbf{\epsilon}_1}  & {\color{blue}\mathbf{\epsilon}_2}  & * & \cdots & *\\
        * & \cdots & * & {\color{blue}\mathbf{\epsilon}_3 }     & {\bf \color{red}\mathbf{v}}  & {\color{blue}\mathbf{\epsilon}_4}  & * & \cdots & *\\
        * & \cdots & * & {\color{blue}\mathbf{\epsilon}_5 }     & {\color{blue}\mathbf{\epsilon}_6}  &  {\bf \color{black}\mathbf{w}} & * & \cdots & *
       \end{array}
 \right)\;,
\end{equation*} where $\mathbf{u}, \mathbf{v},\mathbf{w}$ indicate the three dominant peaks.

\section{The Method}
\subsection{Lorentz Function and Its Sharpening}
From analytic chemistry \cite{nmr}, we learned that an NMR spectrum is represented as a sum of symmetrical, positive valued, Lorentzian shaped peaks, that is the spectral components of an NMR spectrum are Lorentz functions as shown in Fig. \ref{lorentz example}.  Therefore, the NMR spectrum consists of weighted sum of lorentz functions in the following form
\begin{equation*}
  \mathcal{L}(x) = \frac{\left(\frac{1}{2}\Gamma \right )^2h}{(x-x_0)^2+\left( \frac{1}{2}\Gamma\right)^2 }
\end{equation*}
where $ \Gamma$, the scale parameter which specifies its full width at half maximum (FWHM), $x_0$ is the center of the peak, and $h$ is the height.  Apparently the function reaches its maximum height $h$ at $x = x_0$.  For the purpose of analysis, we shall consider the case of $x_0 = 0$ (since one can simply translate the function to achieve the Lorentzian curve at the desired center), in the form  $\displaystyle \mathcal{L}(x) = \frac{w^2h}{x^2+w^2 } $, where $w = \frac{1}{2}\Gamma $, the half width at half maximum (HWHM).  Below are its first several derivatives
\begin{eqnarray*}
\mathcal{L}(x)  & = & \frac{w^2h}{x^2+w^2 }\;, \\
\mathcal{L}'(x) & =  & -2w^2h\frac{x}{(x^2+w^2)^2}\;,\\
\mathcal{L}''(x) & =  & 2w^2h\frac{3x^2-w^2}{(x^2+w^2)^3}\;, \\
\mathcal{L}^{(3)}(x) & = & 2w^2h \frac{12x(w^2-x^2)}{(x^2+w^2)^4}\;
\end{eqnarray*}
and their graphs shown in Fig. \ref{graphs}.  We consider the function $ D(x) =\mathcal{L}(x) -\mathcal{L}''(x) $ and get an idea how the peak in $D(x)$ is sharper than $\mathcal{L}(x)$
\begin{eqnarray}
\label{sharp}
D(x) & = &\mathcal{L}(x)  -\mathcal{L}''(x)\\
              & = & \frac{w^2h}{x^2+w^2 } -2w^2h\frac{3x^2-w^2}{(x^2+w^2)^3} \\
              & = &w^2h \frac{x^4+ 2(w^2-3)x^2 + w^4+2w^2}{(x^2+w^2)^3}
\end{eqnarray}
As shown in left plot in Fig. \ref{Lor1}, a slightly enhanced signal is achieved as a result of cancelation in the side regions and reinforcements in the center region.  For the data analysis and application, the sharpened curve needs to be nonnegative for all the $x$ values.  We shall next investigate under what condition the sharpened signal $D(x)$ remains nonnegative.  The following theorem offers a lower bound of $w$ for the nonnegativity of $D(x)$.
\begin{theorem}
\label{lower}
The sharpened signal $D(x) =\mathcal{L}(x)  -\mathcal{L}''(x)$ is nonnegative for all values of $x$ if and only if $w^2\geq \frac{9}{8}$ (or $w\geq \frac{3}{2\sqrt{2}}$).
\end{theorem}
\begin{proof}
Before we get into the proof.  We notice the function $\mathcal{L}''(x)$ has three critical points $x = 0, x = \pm w$ (the zeros of $\mathcal{L}^{(3)}(x)$ ) and it attains absolute minimum value $-\frac{2h}{w^2}$ at $x = 0$, and maximum $\frac{h}{2w^2}$ at $x = \pm w$.  It can also be seen that function $\mathcal{L}(x)$ has the absolute maximum at $x = 0$, so the sharpened signal $D(x)$ achieves its maximum value $h(1+\frac{2}{w^2})$ at $x = 0$.  We define $\alpha =1+\frac{2}{w^2} $  as the sharpening factor, clearly a bigger $\alpha$ means a better sharpening.  This also implies that the sharpening is less noticeable for wider peaks (bigger $w$).

In order for $D(x) = \mathcal{L}(x)  -\mathcal{L}''(x) $ to be nonnegative only if its numerator part $N(x) = x^4+ 2(w^2-3)x^2 + w^4+2w^2 \geq 0  $ (because its denominator is always positive).  Then the problem is to determine for what values of $w$, $N(x)\geq 0$.  First of all if $w^2\geq 3, x^4+ 2(w^2-3)x^2 + w^4+2w^2\geq 0$.  Now we investigate the case when $w^2<3$, consider the derivative of $N(x)$,
\begin{equation*}
N'(x) = 4x^3-4(3-w^2)x= 0
\end{equation*} solves for the three critical points of $N(x)$, $x = 0; x = \pm \sqrt{3-w^2}$.  By the first order derivative test,  $N(x)$ attains its minimum at $x = \pm \sqrt{3-w^2}$ (symmetry),
\begin{eqnarray*}
N(\pm \sqrt{3-w^2})& = & (3-w^2)^2-2(3-w^2)(3-w^2) + w^4 +2w^2\\
                                     & = & -(3-w^2)^2+ w^4+2w^2 \\
                                     & = & 8w^2-9\;.
\end{eqnarray*}
Therefore, $N(x)$ will remain nonnegative if $8w^2-9 \geq 0$.   We conclude that if $w^2\geq \frac{9}{8}$ (or $w\geq \frac{3}{2\sqrt{2}}$) then  $D(x) =\mathcal{L}(x)  -\mathcal{L}''(x)  $ is nonnegative for all values of $x$.
 \end{proof}

\begin{figure}
\includegraphics[height=7cm,width=14cm]{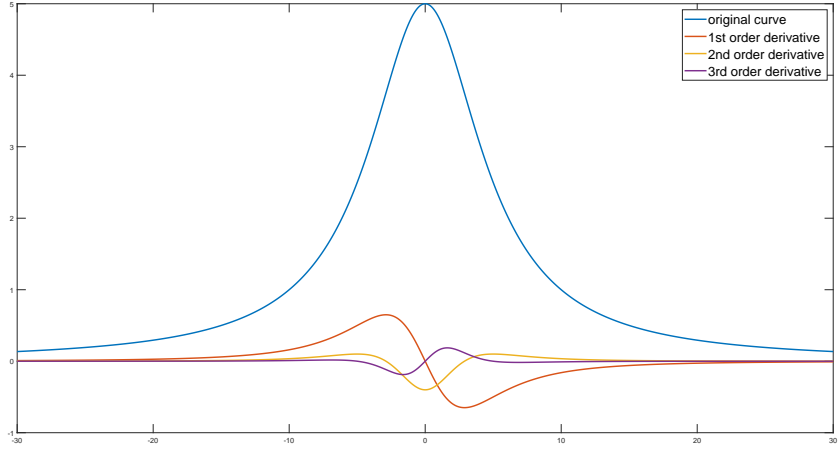}
\caption{The values of the parameters are :  $w = 5, h = 5$. }
\label{graphs}
\end{figure}
\begin{rem}
  Note that the wider the peaks, the less noticeable sharpening will be achieved since the sharpening factor $\alpha =1+  \frac{2}{w^2} $ is close to 1 for wide peak signals.  In order to achieve a recognizable and better sharpening for such signals, we shall consider a weighted sharpening below.
\end{rem}

The weighted peak sharpening defined as
\begin{equation}\label{wtsh}
D_k(x) =\mathcal{L}(x)  - k \mathcal{L}''(x)= w^2h \frac{x^4+ 2(w^2-3k)x^2 + w^4+2k w^2}{(x^2+w^2)^3}
\end{equation} where the weight $k>0$.  The sharpening factor $\alpha = 1+ k \frac{2}{w^2}$.  An immediate question is to the find the optimal value for $k$ to achieve the best balance of sharpening and flatness of the line (nonnegativity).  We have the following result,
\begin{theorem}
\label{upper}
 The upper bound value of $k$ for the weighted sharpening defined in Eq. (\ref{wtsh}) is $w_{\mathrm{opt}} = \frac{8}{9}w^2 $, in which case the sharpening factor is $\alpha  =  \frac{25}{9}$.
\end{theorem}
\begin{proof}
Following the similar argument in the proof of Theorem \ref{lower}, it is clear that if $w^2
\geq 3k, (k\leq \frac{w^2}{3})$, the term $N_k(x) = x^4+ 2(w^2-3k)x^2 + w^4+2k w^2 \geq 0$.  If $k> \frac{w^2}{3}$,
the zeros of $N_{k}'(x) = 4x^3-4(3k-w^2)x $ are $x = 0, x = \pm \sqrt{3k-w^2}$.  $N_k(x)$ obtain its absolute minimum at $x = \pm \sqrt{3k-w^2}$; $N_k(\pm \sqrt{3k-w^2}) =-(3k-w^2)^2 + w^4+2kw^2 = 8kw^2-9k^2 $. Hence $N_k(x)$ will be always nonnegative if $8kw^2-9k^2\geq 0$ or $k\leq \frac{8}{9}w^2$.   The optimal choice is $k_{\mathrm{opt}} = \frac{8}{9}w^2$ for the best sharpening enhancement, and the sharpening factor is $\alpha = 1+\frac{8}{9}w^2
\cdot \frac{2}{w^2} =  \frac{25}{9}$ which means that the sharpened peak is about 2.8 times higher yet narrower.  Please be noted that the value of $k$ is user preset and can be any number between 1 and $k_{\mathrm{opt}}$.
\end{proof}
The sharpening effects are depicted in Fig. \ref{Lor1}, the first plot shows a one-peak signal and the sharpening without weight where it can be seen that the sharpening is barely noticeable comparing to the original signal; while the second plot shows that the better performance by the weighted sharpening.  More plots in Fig. \ref{Lor2} demonstrate the results of weighted sharpening of a multi-peak signal as well as superimposition of multiple signals.
\begin{figure}
\includegraphics[height=4cm,width=8cm]{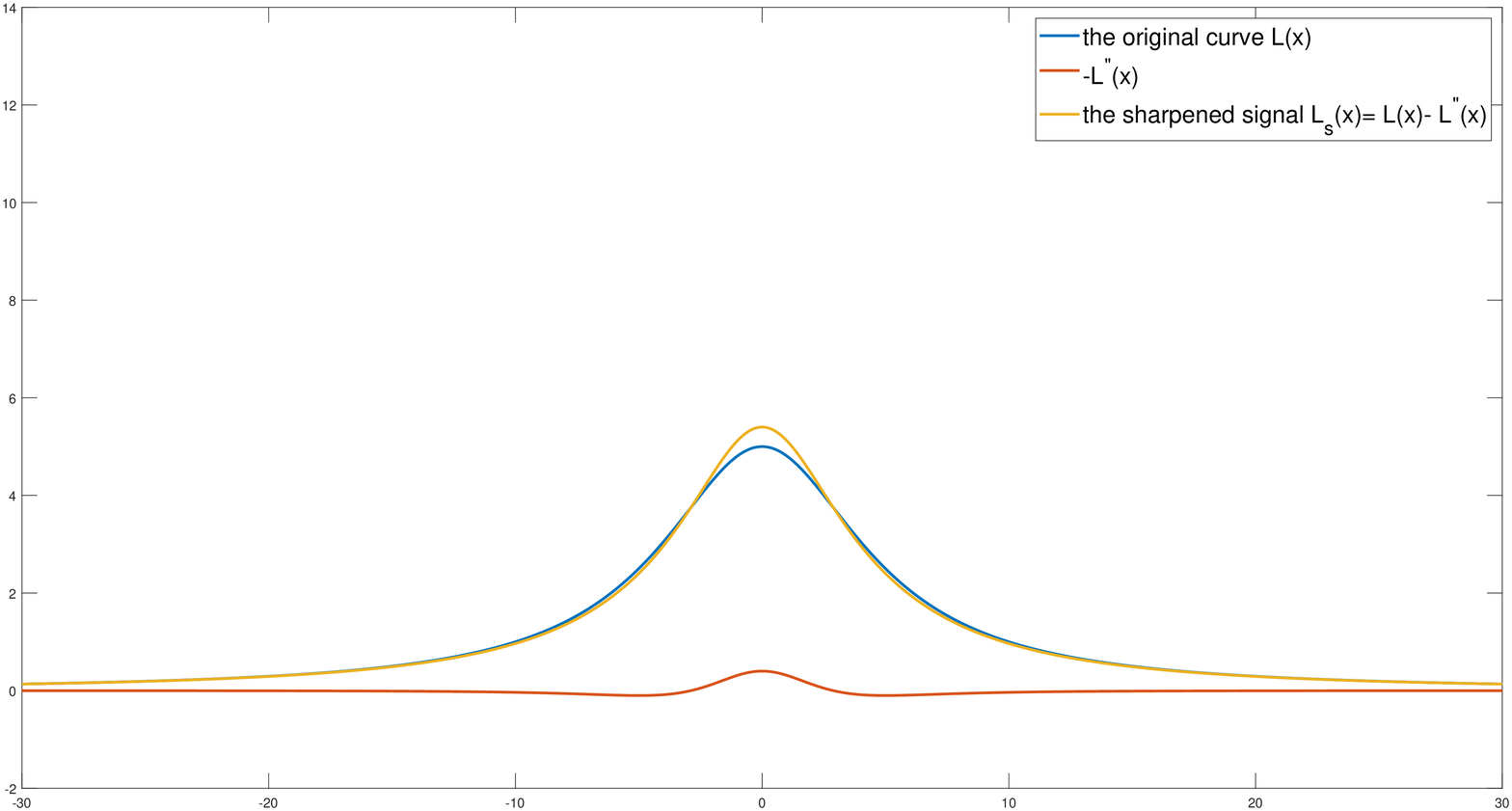}
\includegraphics[height=4cm,width=8cm]{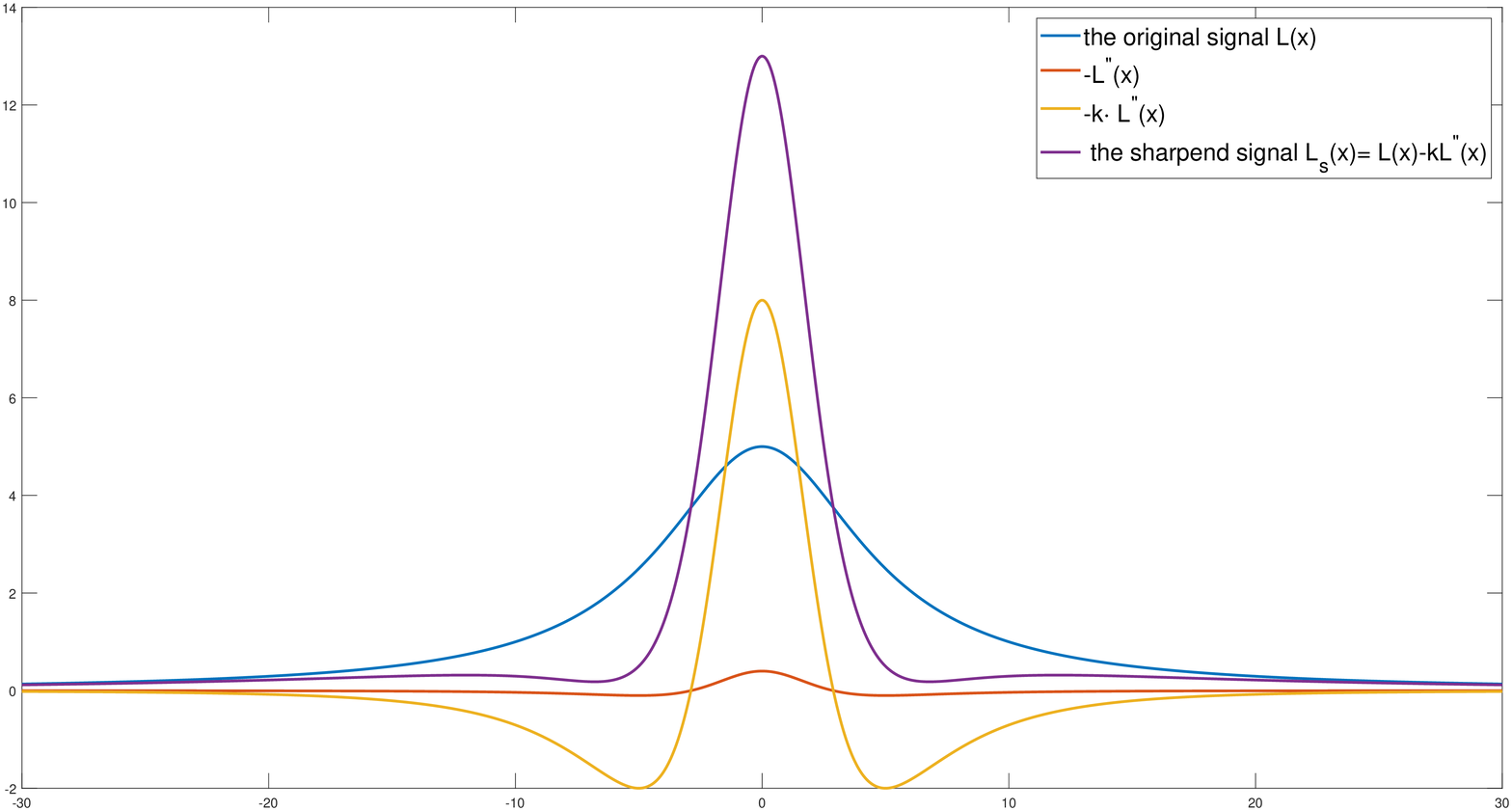}
\caption{Left panel shows a signal and sharpening without weight (or $k = 1$); Right panel is the same signal and weighted sharpening with $ k = 20$. Other parameters are $w = 5, h = 5$. }
\label{Lor1}
\end{figure}
\begin{figure}
\includegraphics[height=4cm,width=8cm]{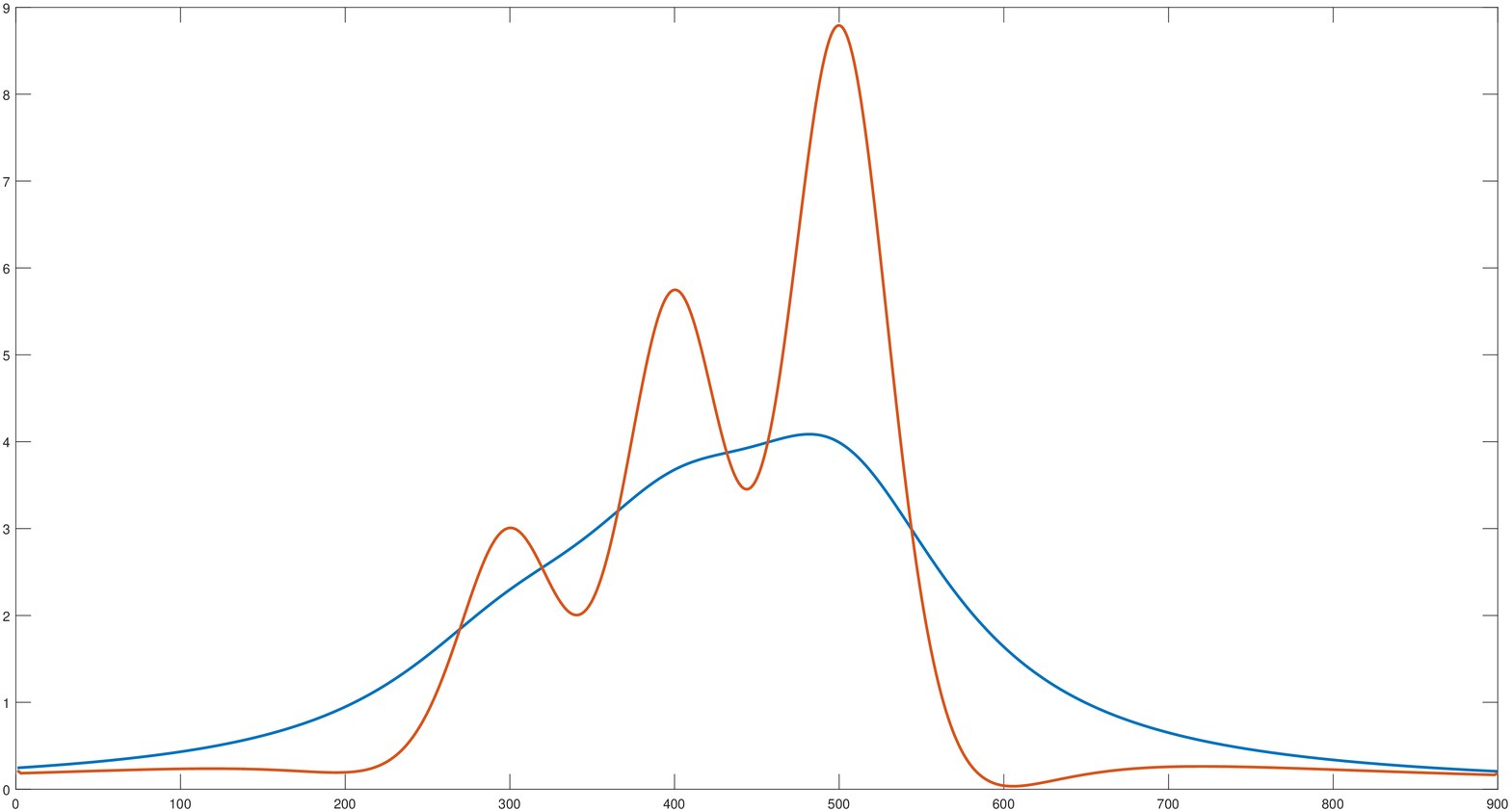}
\includegraphics[height=4cm,width=8cm]{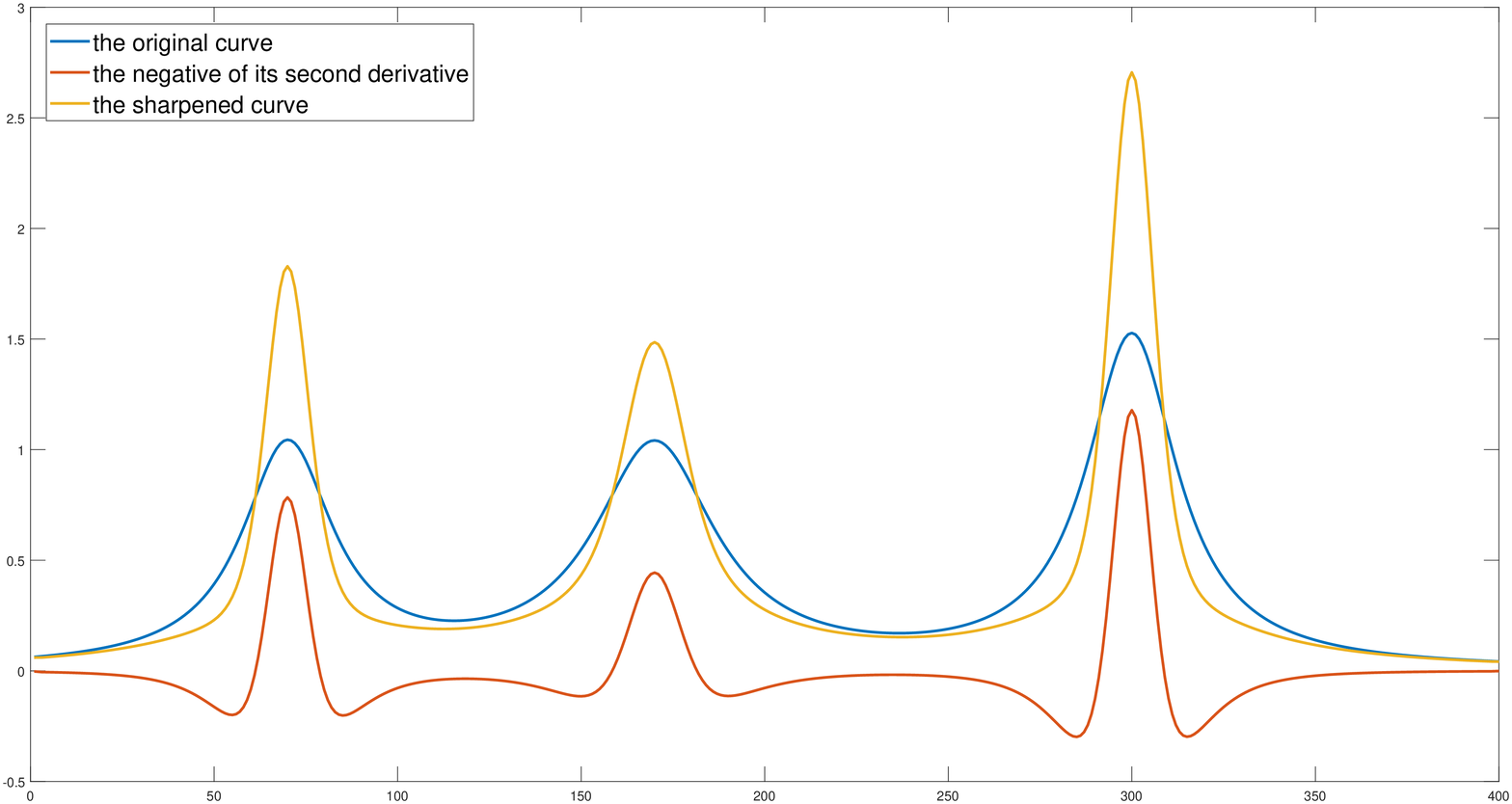}
\caption{Left panel shows a signal with three Lorentzian peaks; the negative of its second derivative, and the sharpening of its peaks. Right panel is a mixed signal formed from a combination of three signals. }
\label{Lor2}
\end{figure}

\subsection{Mixed Signal Sharpening and Separation}
We shall make the following definition
\begin{defi}
For a given signal $s(x)$, the weighted sharpening operator $\mathcal{P}$ is defined as
$\mathcal{P}s(x) = s(x)-ks''(x) $, where $k>0$ is the user preset weight parameter.
\end{defi}
The linearity of the operator follows from $ \mathcal{P}\left(as_1(x) + bs_2(x)  \right) = a\mathcal{P}s_1(x)+ b\mathcal{P}s_2(x)$.

Consider the linear mixture model equation (\ref{BSS}) $X = A S$, where $X\in \mathbb{R}^{m\times p}, A \in \mathbb{R}^{m\times b}, S \in \mathbb{R}^{n\times p} $. Rows of $X$ represents the measured spectral mixtures,
and rows of $S$ are the source signals. Matrix $A$ contains the mixing coefficients.  Each row $X_i$ can be expressed as
\begin{equation*}
  X_i = \sum_{j = 1}^{m}a_{ij}S_j\;.
\end{equation*}
Then we apply the weighted sharpening operator on $X_i$
\begin{equation*}
\mathcal{P}X_i = \mathcal{P} \sum_{j = 1}^{m}a_{ij}S_j = \sum_{j = 1}^{m}a_{ij}\mathcal{P} S_j
\end{equation*}
By the previous discussion, $\mathcal{P} S_j$ is the $j$th sharpened signal with narrower peaks of enhanced resolution than $S_j$, then the dominant peak condition is clearly much better satisfied.  We shall sharpen all the mixed signals (all the rows of matrix $X$) to have the following preprocessed data (which can be formally written)
\begin{equation*}
  \hat{X} = \mathcal{P} X = \mathcal{P}AS  = A\hat{S}\;,
\end{equation*}
$\hat{S} = \mathcal{P}S$ each row of which represents a sharpened source signal.   For the half width at half maximum parameter $w$ used in the selection of the weight $k$ (since $k_{\mathrm{opt}} = \frac{8}{9}w^2$) for numerical implementation, an estimate of the narrowest peak width suffices.  One can read off approximate value from mixture signals if the dominant interval(s) happen to contain a peak.  In more complicated NMR data, the expertise of an analytical chemist may also be helpful to estimate this
parameter.

Once the rows of the mixture matrix $X$ being preprocessed (peaks sharpened), we then apply the NN method on $\hat{X} = A\hat{S}$ to retrieve the columns of $A$ by solving either problem (\ref{LPNF}) or (\ref{NNCG}).  In the presence of noise,  the following optimization problem is suggested to solve for an estimate the mixing matrix $A$ {\allowdisplaybreaks
\begin{equation}
\label{LPNP2}
\mathrm{score} =\min_{\lambda_j \geq 0} \frac{1}{2} \|\sum^{p}_{j = 1, j\neq
k}\hat{X}^{j} \lambda_j - \hat{X}^{k}
 \|^2_2\;, k = 1,\dots,p\;\\
  \;,
\end{equation}
which can be solved by nonnegative least squares method.  A column with a low
score is unlikely to be a column of $A$ because this column is
approximately a nonnegative linear combination of the other columns of
$X$; while a high score may suggest that the corresponding
column is far from being a nonnegative linear combination of other
columns of $X$.  In practice, the $n$ columns from $X$ with highest
scores will be selected as an estimate of $A$.  In NN method, the Moore-Penrose inverse $A^{+}$ of $A$ is computed and used to obtain an estimate of the source signal $S$: $S = A^{+}X $.  The recovered $S$ might contain negative values due to the error in the estimate of $A$.  For a remedy, if $m\geq n$ (over-determined), then a nonnegative least squares method can be adopted for solving the source matrix $S$; for each column $S^i$ of $S$, solve the problem
\begin{equation*}
 \min_{S^i\geq 0}\frac{1}{2} \|X^i - AS^i \|^2_2\;.
\end{equation*}
If $m<n$ (under-determined), the solution of $S$ is non-unique, but one can solve a nonnegative $\ell_1$ optimization problem for a sparse solution of $S^i$,
\begin{equation}
\label{l1BSS}
 \min_{S^i\geq 0}\frac{1}{2}\|X^i - AS^i \|^2_2 + \mu \|S^i\|_1\;.
\end{equation}
We shall assign a tiny value to $\mu$ when there is minimal measurement error to heavily weigh the term $\|X^i - AS^i \|^2_2$ so that $X^i = AS^i $ is nearly satisfied.  To solve (\ref{l1BSS}) , we may use linearized Bregman method \cite{G_O_09,YO} with a proper projection onto nonnegative convex set.

\section{Numerical Experiments}
We report in this section the numerical results of the proposed method.  Hereafter, NN method is the convex cone method without sharpening preprocessing, while the term NNP method is NN method with peak sharpening.  First example contains synthetic data, there are two mixture and two source signals ($m = n = 2$).  The source spectra are synthesized using Lorentzian shaped peaks to mimic the real NMR spectra, the mixture matrix are generated by the model $X = AS$.   Fig. \ref{ex1Source} shows the source spectra, while the left panel in Fig. \ref{ex1mix} is a mixed signal and its sharpening.  We also show the scattered cloud of the columns of $X$ in the right panel of  Fig. \ref{ex1mix}.  The recovered source spectra by NN method and NNP method are depicted in Fig. \ref{Ex1SourceNN}.  Both methods recovered the source signals rather well comparing to the ground truth.  The estimate of
mixing matrix $A_{\mathrm{NN}}$ by NN method, $A_{\mathrm{NNP}}$ by NNP, and the true mixing matrix $A_{\mathrm{TR}}$ (noted that first rows of all matrices are scaled to the same for the purpose of illustration) are shown and compared below.
\begin{equation*}
  A_{\mathrm{TR}} =\left(
   \begin{array}{cccc}
    0.6  &  0.8\\
    0.8  &  0.6
   \end{array}
 \right)\;,\;
A_{\mathrm{NN}} =\left(
   \begin{array}{cccc}
    0.6 &   0.8\\
    0.7478 &   0.6427
    \end{array}
    \right)\;,\;
    A_{\mathrm{NNP}} =\left(
   \begin{array}{cccc}
    0.6 &  0.8\\
    0.7890  &  0.6085
    \end{array}
    \right)\;
 \end{equation*}

To compare the performance of mixing matrix estimates, we calculate Comon's index defined here.
\begin{defi}
Consider two nonsingular matrices let $A$ and $\hat{A}$ with normalized columns.  The distance between $A$ and $\hat{A}$ denoted by $\varepsilon(A,\bar{A})$ which is
\begin{equation*}
\varepsilon(A,\bar{A}) = \sum_i\biggl | \sum_j |d_{ij}|-1\biggr |^2 + \sum_j\biggl | \sum_i |d_{ij}|-1\biggr |^2
 + \sum_i\biggl | \sum_j |d_{ij}|^2-1\biggr | + \sum_j\biggl | \sum_i |d_{ij}|^2-1\biggr |\;,
\end{equation*}
where $D =A^{-1}\bar{A} $, and $d_{ij}$ is the entry of $D$.
\end{defi} Comon proved in \cite{Comon} that
$A$ and $\bar{A}$ are considered nearly equivalent in BSS problems  if $\epsilon(A,\bar{A}) \approx 0$.  We computed Comon's index between the true mixing matrix and estimates by NN method and NNP method
\begin{equation*}
\varepsilon(A_{\mathrm{TR}},A_{\mathrm{NN}}) = 0.8012\;,\;  \varepsilon(A_{\mathrm{TR}},A_{\mathrm{NNP}}) = 0.1818.
\end{equation*}
$A_{\mathrm{NNP}}$ is much closer to $A_{\mathrm{TR}}$, implying a better estimate.  We also studied the relation of the sharpening weight $w$ and separation results : we let $k$ vary from $5$ to $100$, and computed the Comon's indices and showed the curve in Fig. \ref{dependenceONk}.

\begin{figure}
\includegraphics[height=4.7cm,width=16cm]{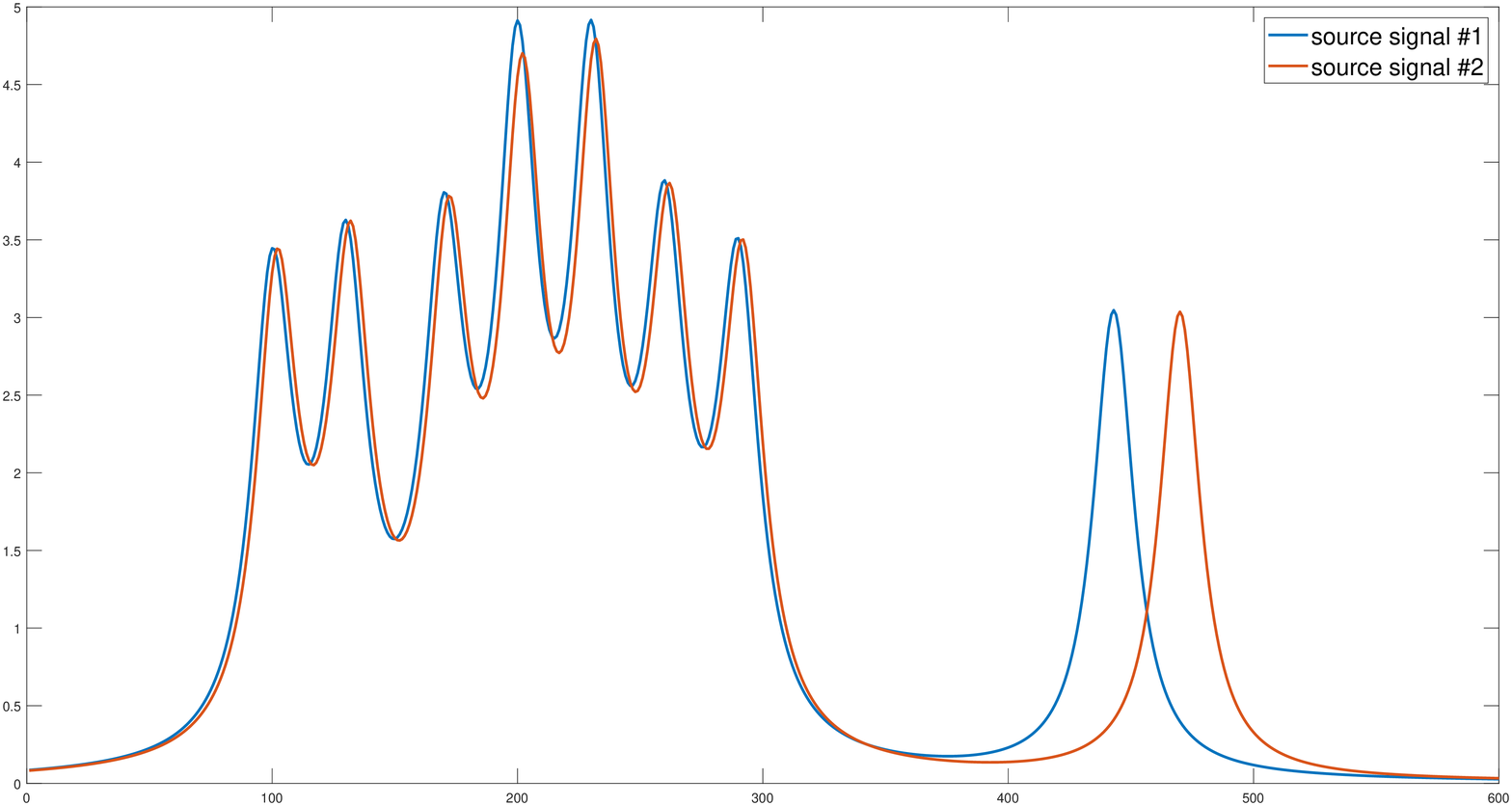}
\caption{The two source signals synthesized in example one. It can be seen that they share majority of their spectral components, the two dominant peaks are located to right side.}
\label{ex1Source}
\end{figure}
\begin{figure}
\includegraphics[height=4cm,width=8cm]{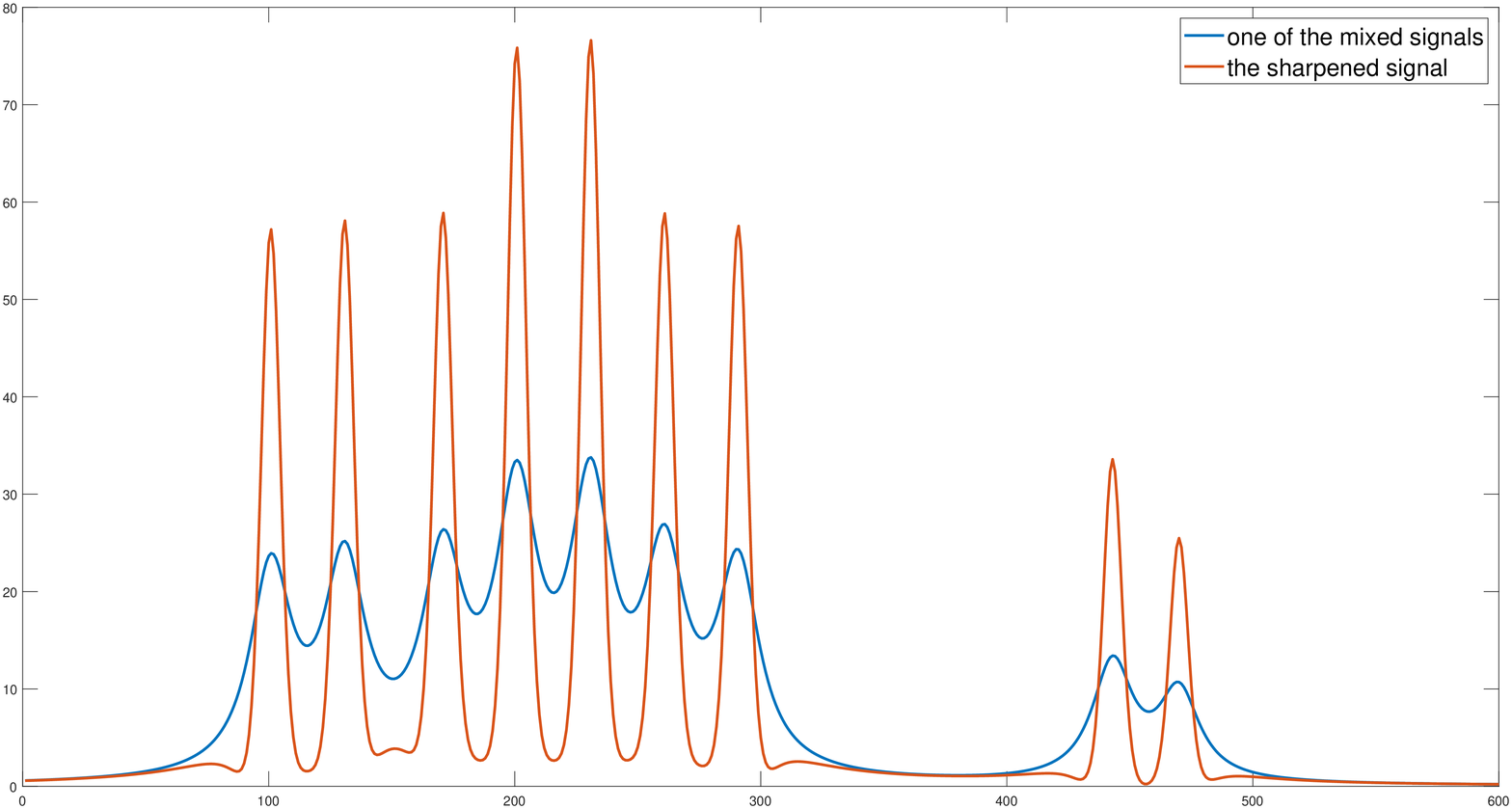}
\includegraphics[height=4cm,width=8cm]{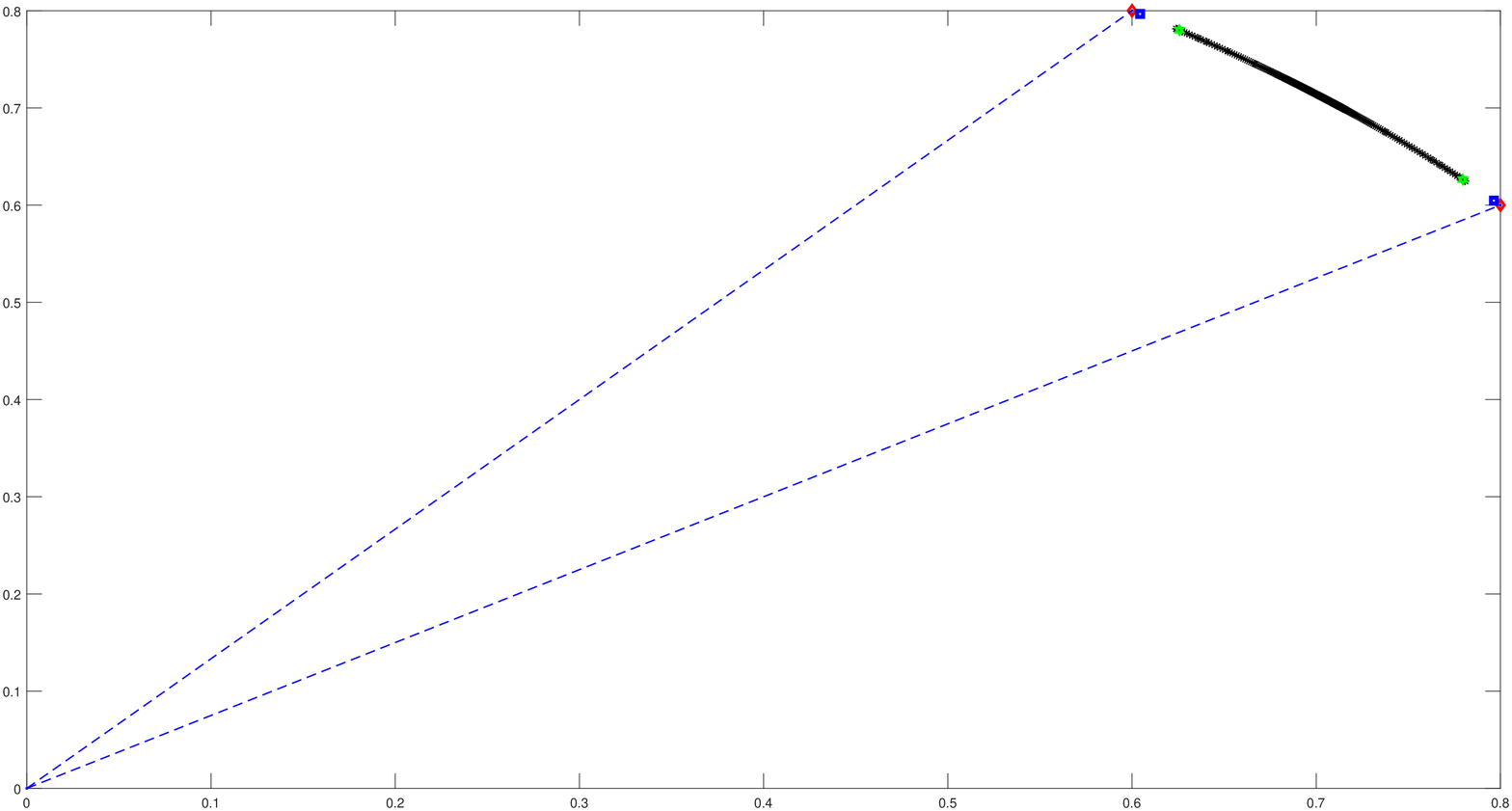}
\caption{Left: a mixed signal and its sharpening. Right: Columns of $X$ indicated by {\bf black stars}.  NN method identified the columns of mixing matrix as the vertices ({\color{green} green triangles}) of a convex cone enclosing the columns of $X$ , while NNP found the {\color{blue}blue circles}.  {\color{red}Red diamonds} represent the columns of true mixing matrix. }
\label{ex1mix}
\end{figure}

\begin{figure}
\includegraphics[height=4cm,width=8cm]{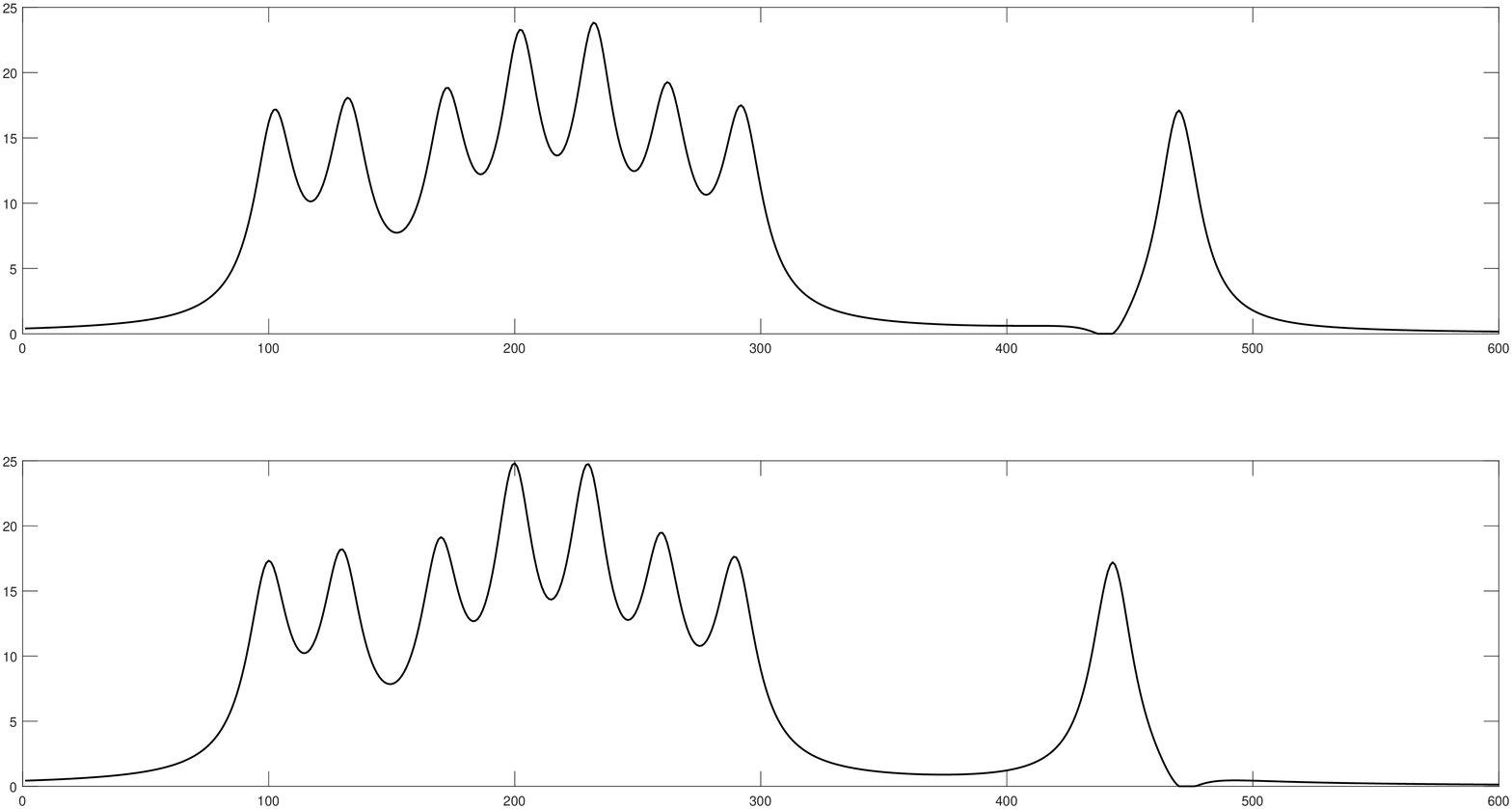}
\includegraphics[height=4cm,width=8cm]{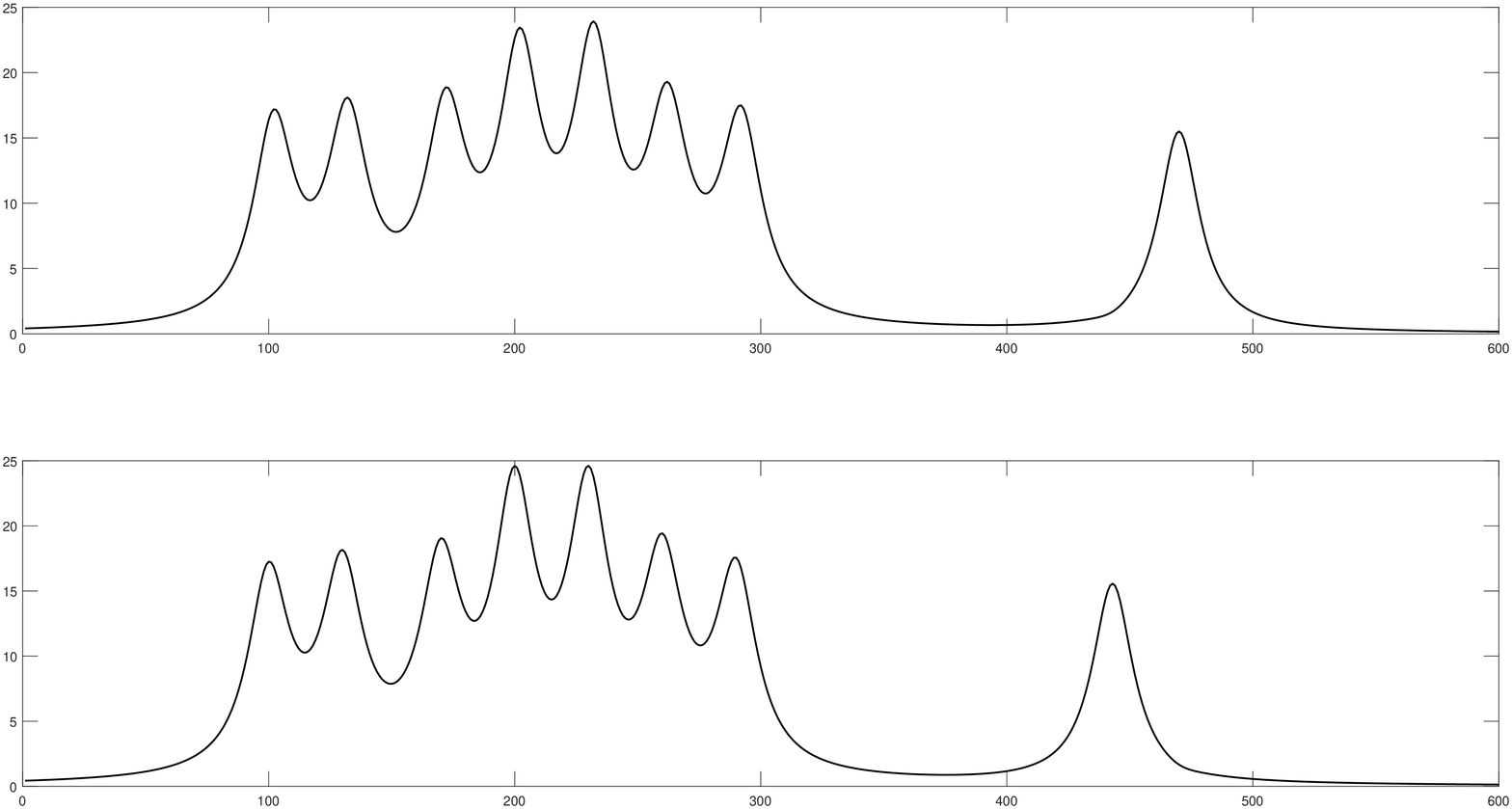}
\caption{left: recovered source signals by NN method; Right: recovered source signals by NNP method. }
\label{Ex1SourceNN}
\end{figure}

\begin{figure}
\includegraphics[height=5cm,width=16cm]{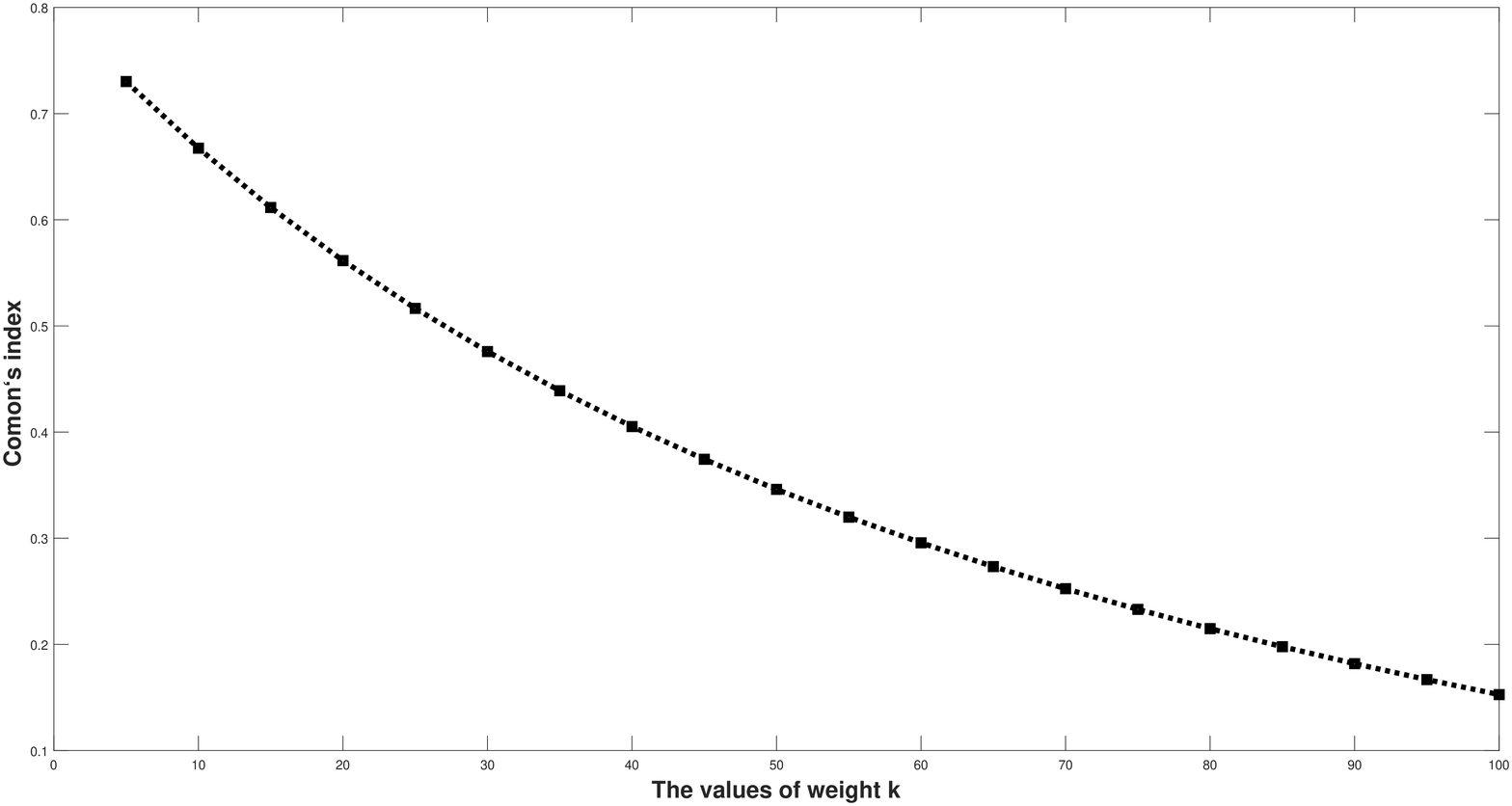}
\caption{Comon's indices v.s. the sharpening weights.}
\label{dependenceONk}
\end{figure}

In the second example, we present the numerical results of three mixtures and three sources signals.  With the concept of Comon's index, we show the robust performances of
NNP method for noisy spectral data.  The three sources signals in Fig. \ref{eg2_1}
were linearly combined to generate three mixtures, and then
Gaussian noises with SNR varying from 30 to 120 dB were added.  Figure \ref{eg2_4} indicates the robustness of our
method with small indices even in the low SNR zone.  The comparison of the recovered mixing matrix by NN method,  NNP method, and the ground truth are shown here
\begin{equation*}
  A_{\mathrm{TR}} =\left(
   \begin{array}{cccc}
   0.6667  &  0.2727 &   0.2000 \\
    0.2222 &   0.4545 &   0.3000\\
    0.1111 &   0.2727 &   0.5000
   \end{array}
 \right)\;,\;
A_{\mathrm{NN}} =\left(
   \begin{array}{cccc}
    0.6667  & 0.2727 &   0.2000\\
    0.2793  &  0.3875 &  0.2904\\
    0.1672  &  0.2688 &   0.4428
      \end{array}
    \right)\;,\;
     \end{equation*}
\begin{equation*}
     A_{\mathrm{NNP}} =\left(
   \begin{array}{cccc}
          0.6667 &   0.2727 & 0.2000\\
           0.2416 &    0.4551 & 0.3039\\
           0.1312 &   0.3044 & 0.5082
 \end{array}
    \right)\;
 \end{equation*}
Comon's index between the true mixing matrix and estimates by NN method and NNP method here
\begin{equation*}
\varepsilon(A_{\mathrm{TR}},A_{\mathrm{NN}}) = 1.3362\;,\;  \varepsilon(A_{\mathrm{TR}},A_{\mathrm{NNP}}) = 0.4952.
\end{equation*}
Clearly the NN method with sharpening preprocessing delivers better results.  Figs. \ref{eg2_1}--\ref{eg2_4} show the computational results for the readers' perusal.  The sharpening parameter we used in this example is $k = 40$ which proves to work well.  It can be seen that both methods are able to capture the peaks and their locations of the source signals as shown in Fig. \ref{eg2_3}, a closer look at the comparison with the real source signal in the left plot of Fig. \ref{eg2_4} clearly shows the better performance of NNP method.  The regions marked by arrows are the discrepancies of the result of NN method with the ground truth.  The similarity between the signals measured by their inner products are calculated $\mathrm{sim}(s_{\mathrm{NN}},s_{\mathrm{TR}})= 0.9767, \mathrm{sim}(s_{\mathrm{NNP}},s_{\mathrm{TR}})= 0.9998 $.

\begin{figure}
\includegraphics[height=4cm,width=8cm]{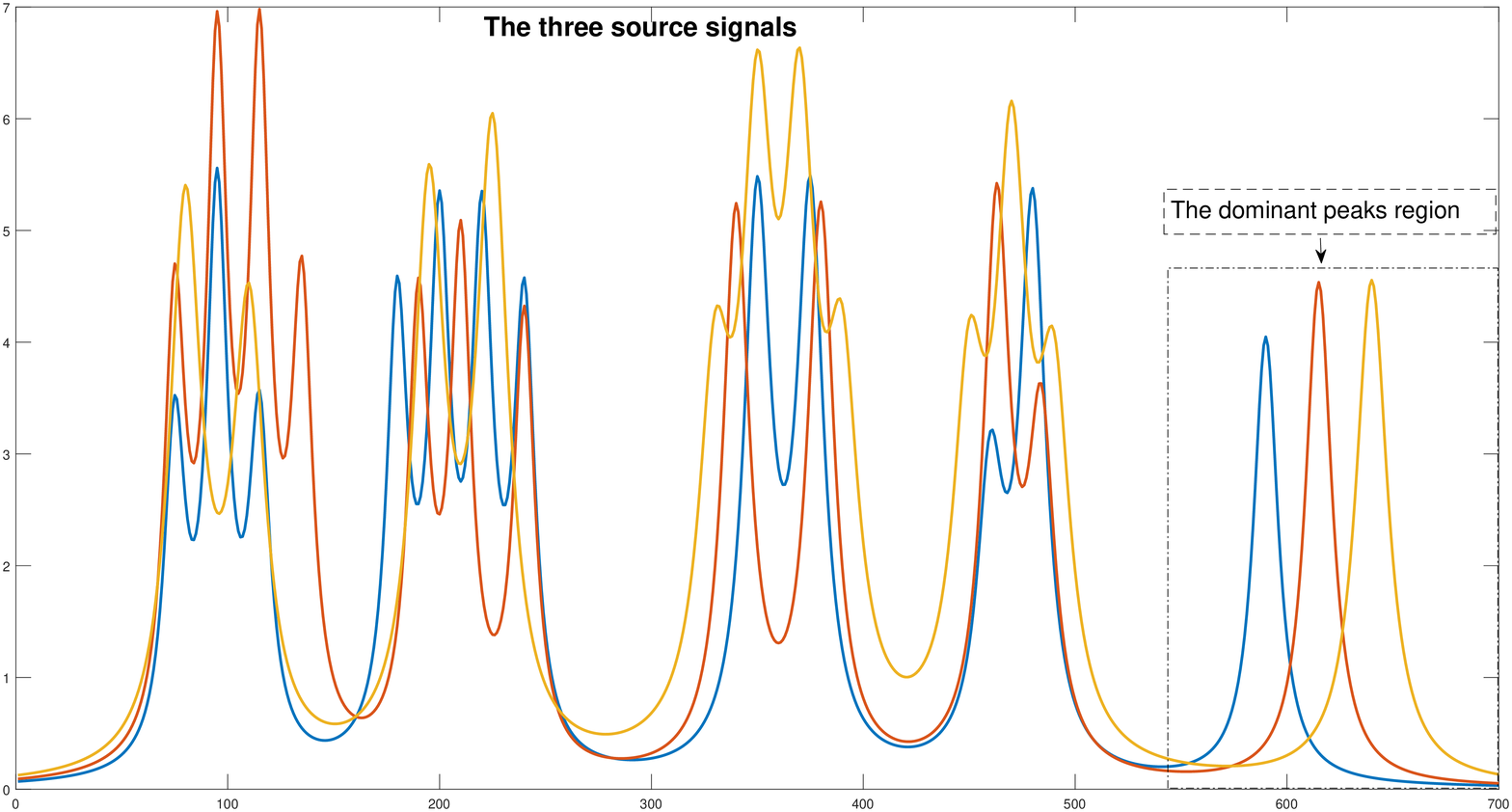}
\includegraphics[height=4cm,width=8cm]{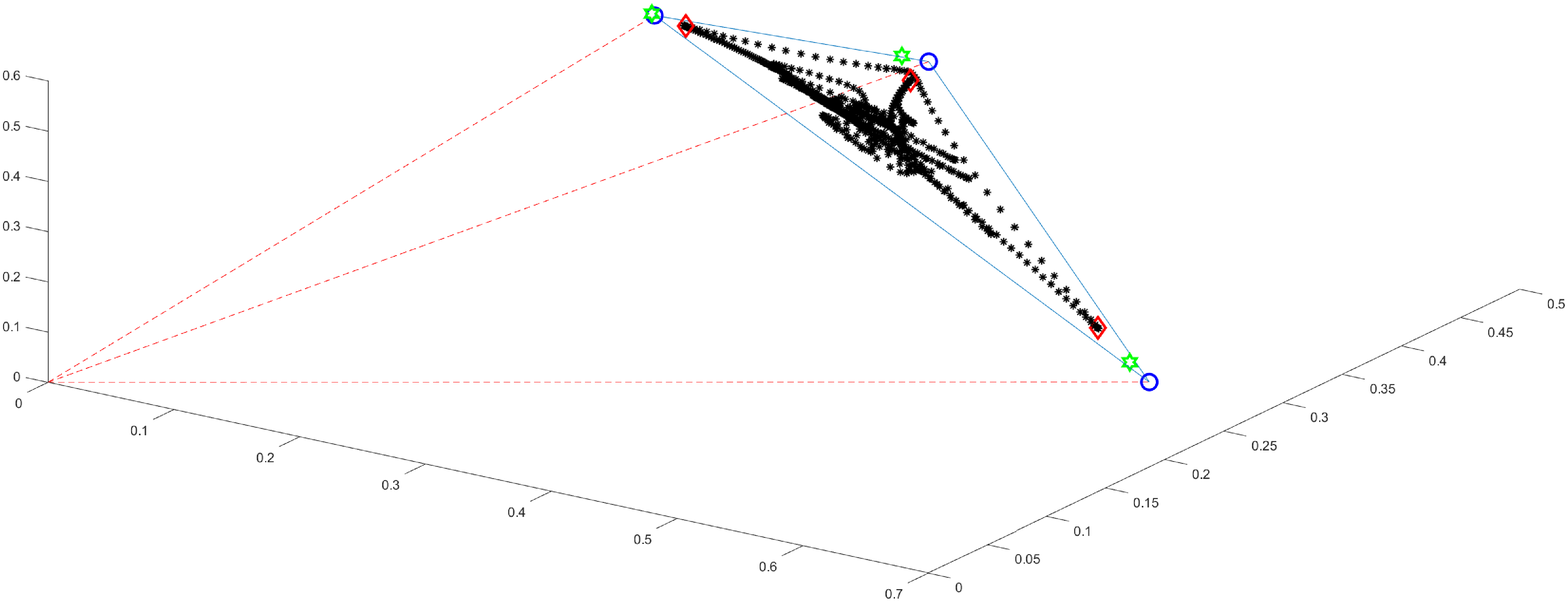}
\caption{Left: Three positive and overlapped Lorentzian source signals with dominant peaks which are shown in the rectangle.  Right: Comparison of
the three columns of mixing matrix recovered from NN with those of the true mixing matrix A (shown in {\color{blue} blue circles}). NN method identifies the columns of mixing matrix as the
edges of a minimal cone enclosing the mixtures (depicted by {\color{red} red diamonds}). The deviation of NN’s results is due to the violation of the condition SAP.  With a preprocessing peak sharpening, NNP method delivers a better results ({\color{green} green stars}) being closer to the {\color{blue} blue circles}).}
\label{eg2_1}
\end{figure}

\begin{figure}
\includegraphics[height=4cm,width=8cm]{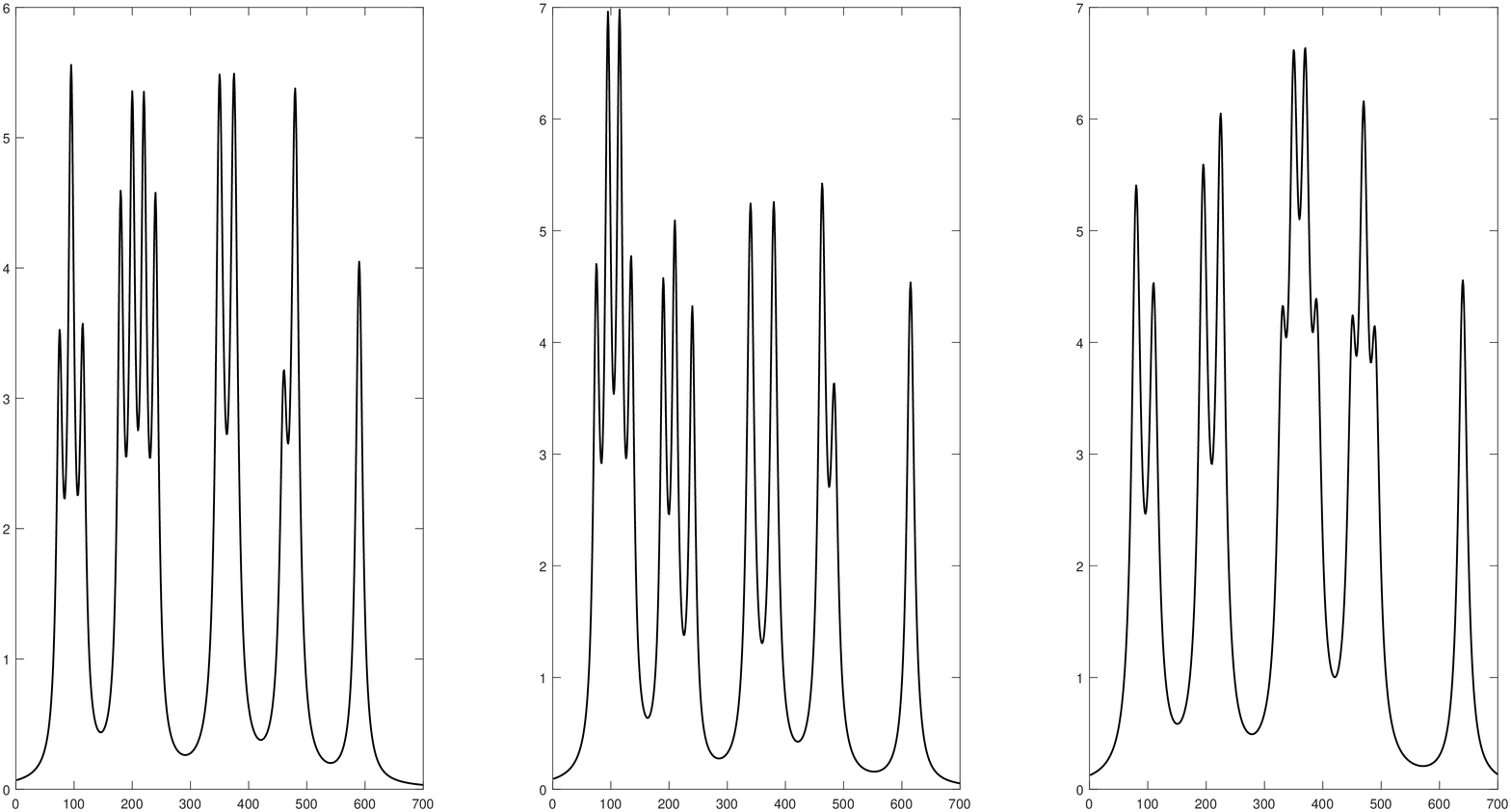}
\includegraphics[height=4cm,width=8cm]{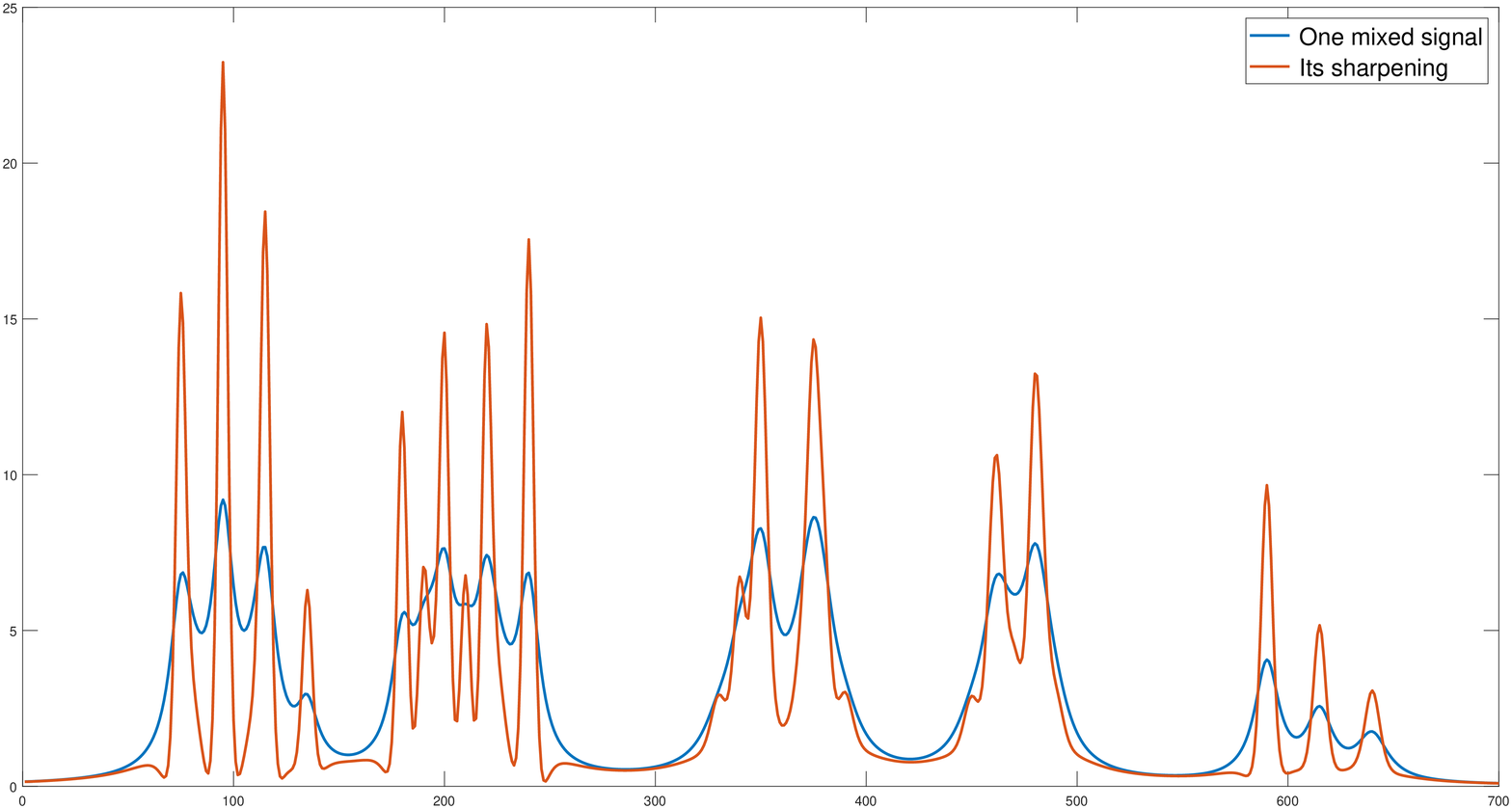}
\caption{Left: the real source signals; Right: one of the signals and its sharpening.}
\label{eg2_2}
\end{figure}

\begin{figure}
\includegraphics[height=4cm,width=8cm]{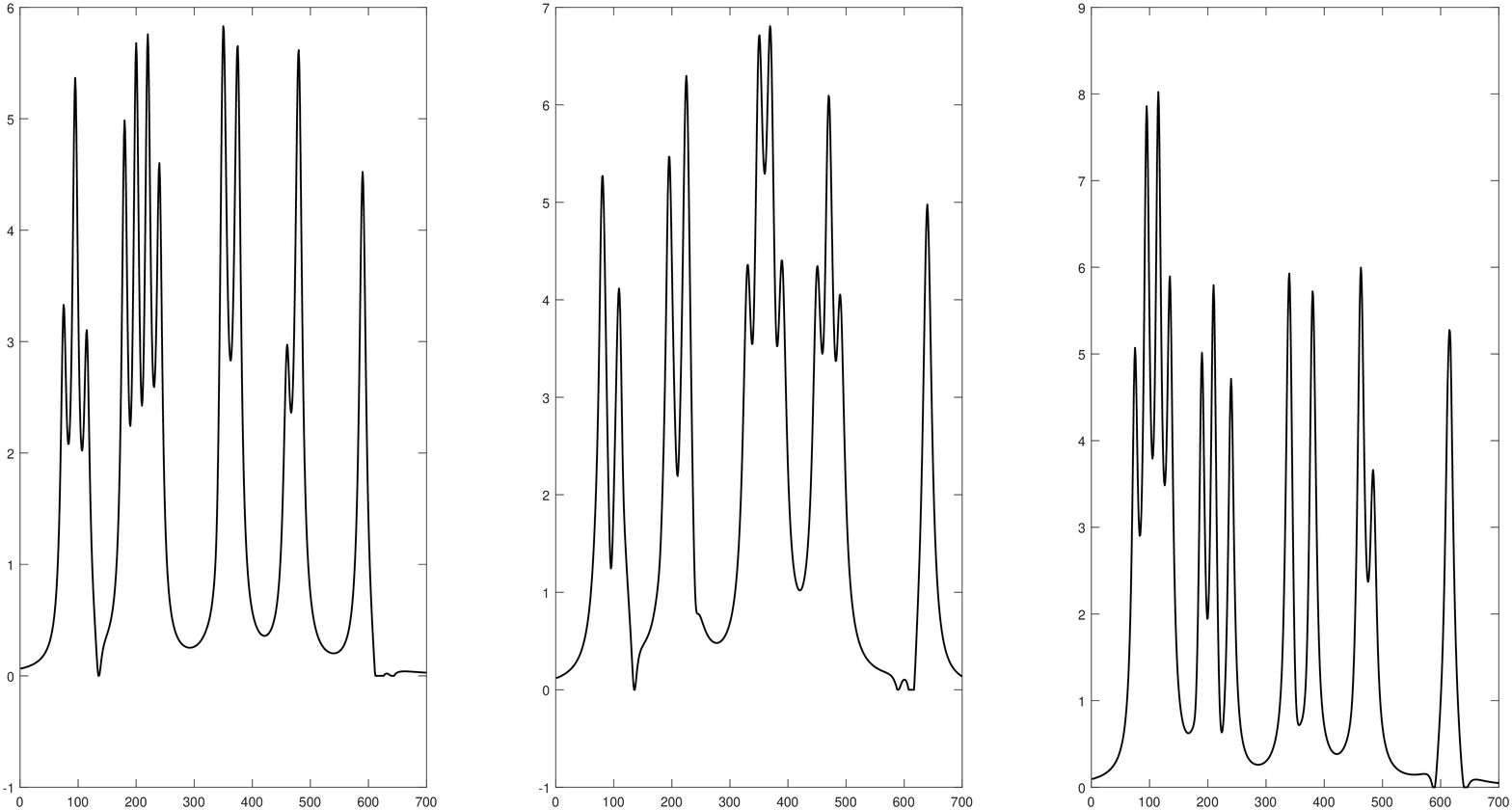}
\includegraphics[height=4cm,width=8cm]{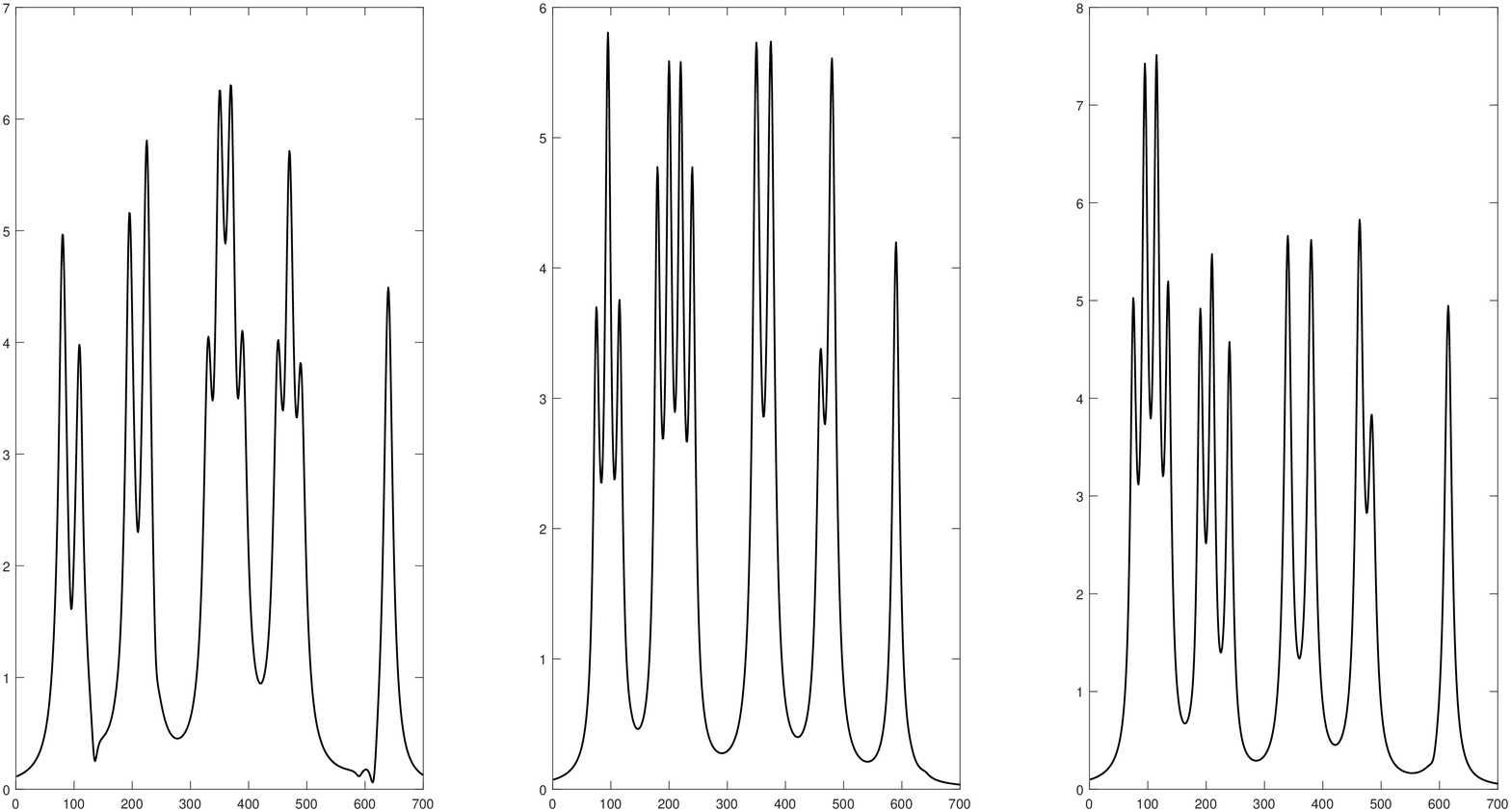}
\caption{Left: the computed source signals by NN method.  Right: the computed source signals by NNP method.}
\label{eg2_3}
\end{figure}

\begin{figure}
\includegraphics[height=5cm,width=9cm]{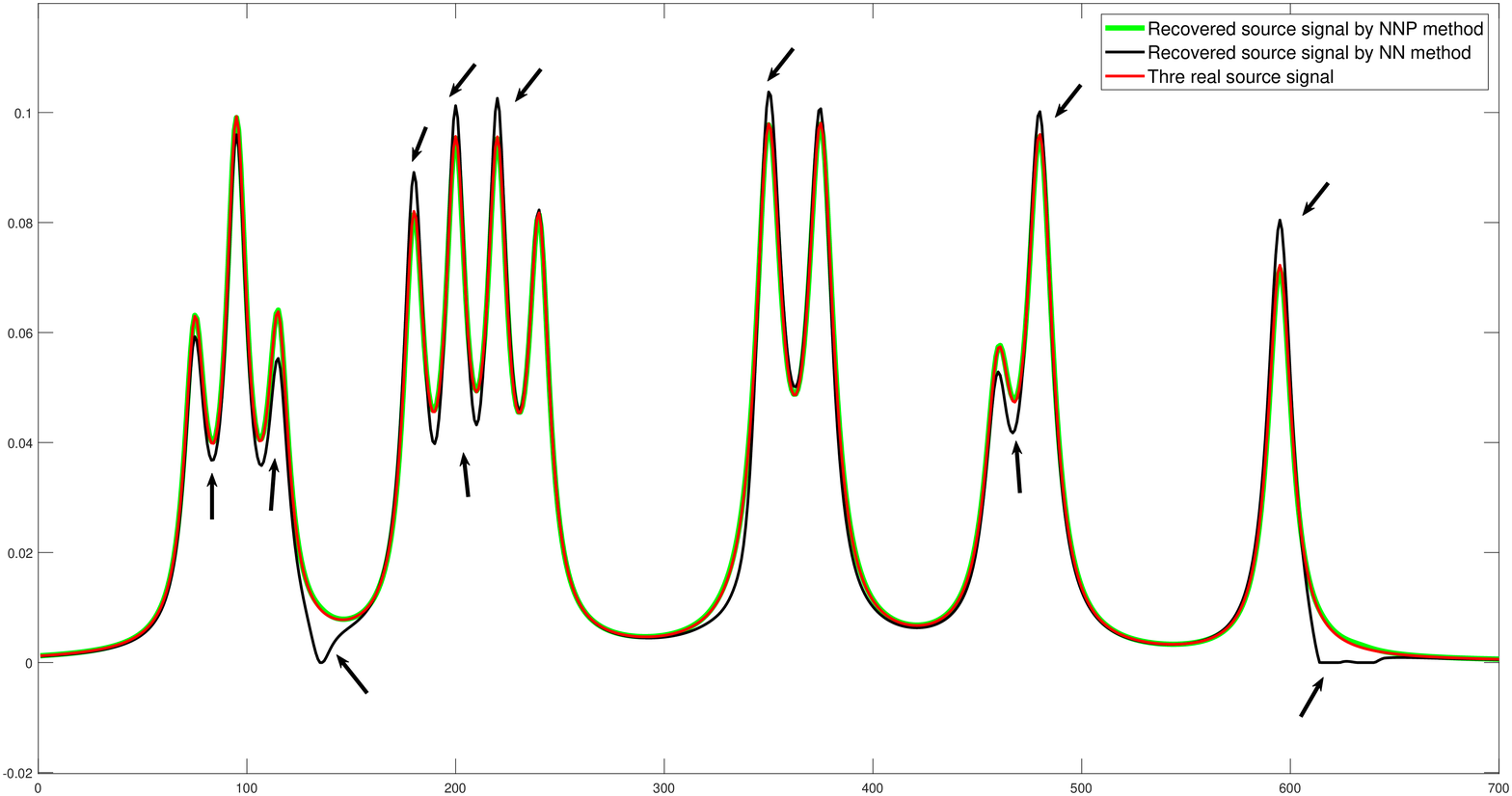}
\includegraphics[height=5cm,width=8cm]{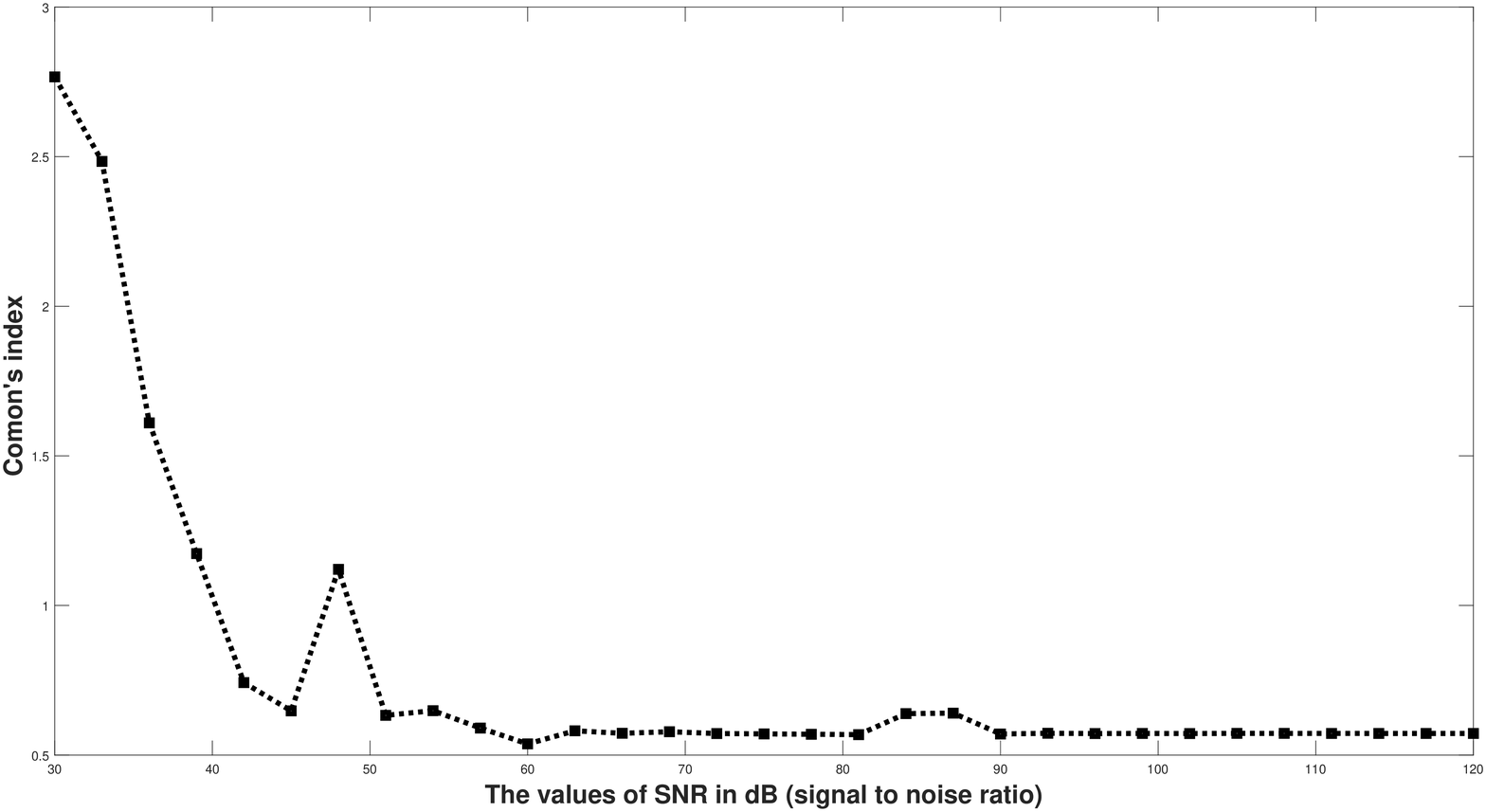}
\caption{Left: The closeness of the recovered source signals and the real one; Right: Robust performance of NNP in the presence of noise.}
\label{eg2_4}
\end{figure}


Next we test the method with real world NMR experimental data.  In Fig. \ref{realNMR} there are three mixtures, each is formed by a linear combination of three 4-peak source signals.  The plot in Fig. \ref{realNMRshp} shows one mixed signal and its sharpening. The three source signals computed by the two methods are shown in the three plots from Fig. \ref{RecNNvsNNP}. Although there are small spurious noisy peaks (or artifacts) around in the results, the four major peaks of signals are well captured and recognizable by both NN method and NNP method.  The second and third plots of Fig. \ref{RecNNvsNNP} show rather similar results by the two methods.  In the first plot, we observe two noticeable bleed through peaks in the signal recovered by NN method, while the two peaks can be barely seen in NNP recovery.
In this example, an estimate of the lower bound of the half peak width $w = 4$ (the narrowest peak) is obtained by examining the mixture signals, we chose the sharpening weight parameter $k = 10$.   Here we also present the recovered mixing matrices by the two methods (note that we do not have ground truth matrix to compare with)
\begin{equation*}
 A_{\mathrm{NNP}} =\left(
   \begin{array}{cccc}
    0.7601 &   0.7454 & 0.8675\\
    0.6481 &   0.3659 & 0.0496\\
    0.0473 &   0.5573 & 0.4949
    \end{array}
    \right)\;,\;
    A_{\mathrm{NNP}} =\left(
   \begin{array}{cccc}
    0.7189  &  0.8741 &    0.7616\\
    0.6952   & 0.0398 &   0.3640\\
         0 &   0.4841 &   0.5362
    \end{array}
    \right)\;
 \end{equation*}
\begin{figure}
\includegraphics[height=4cm,width=14cm]{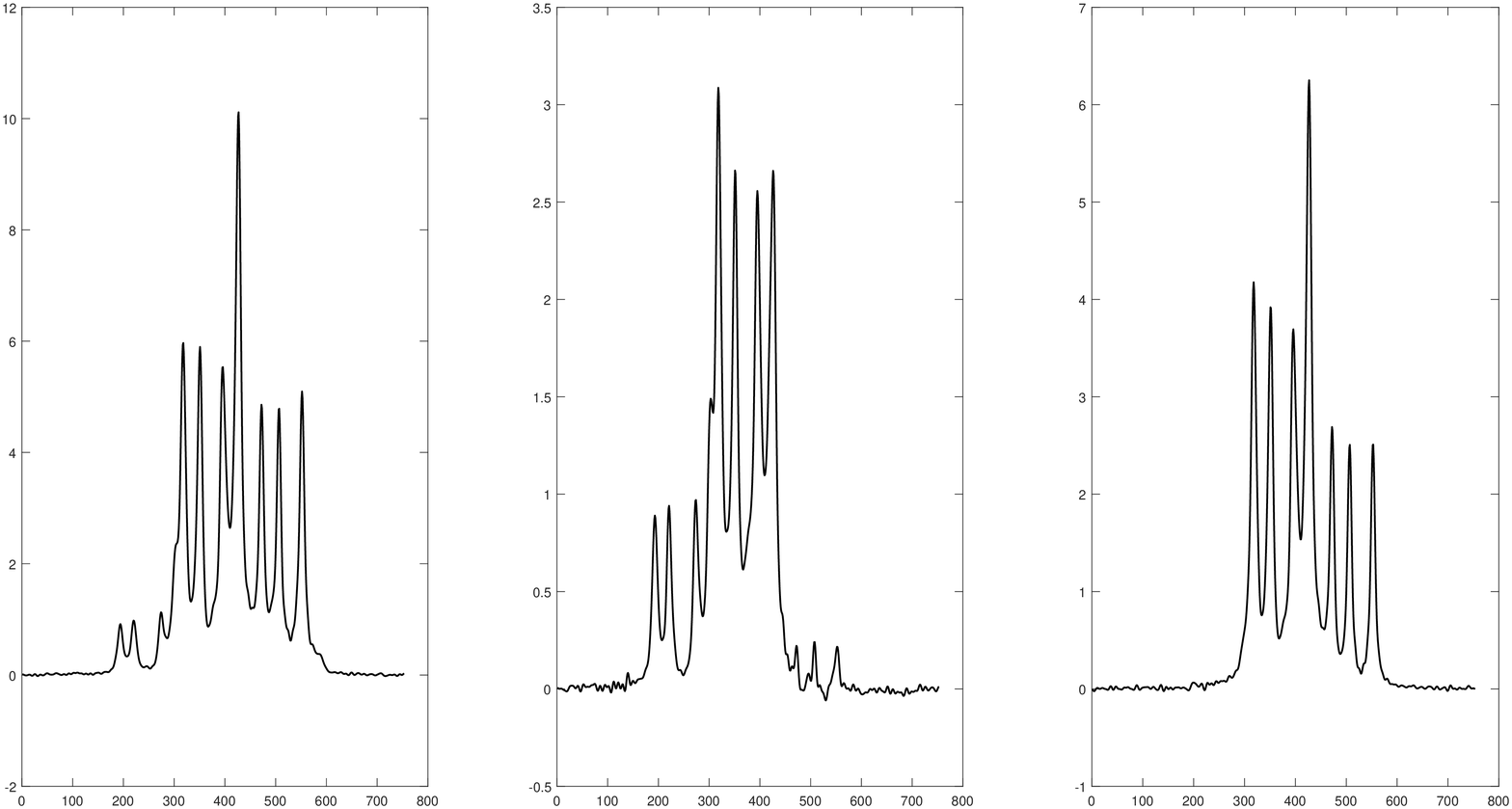}
\caption{Three mixed realistic NMR spectra (from three sources).}
\label{realNMR}
\end{figure}
\begin{figure}
\includegraphics[height=4cm,width=14cm]{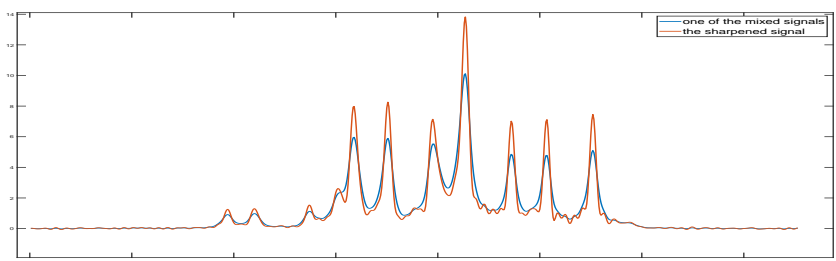}
\caption{One of the mixed NMR signal and its sharpening.}
\label{realNMRshp}
\end{figure}
\begin{figure}
\includegraphics[height=4cm,width=7cm]{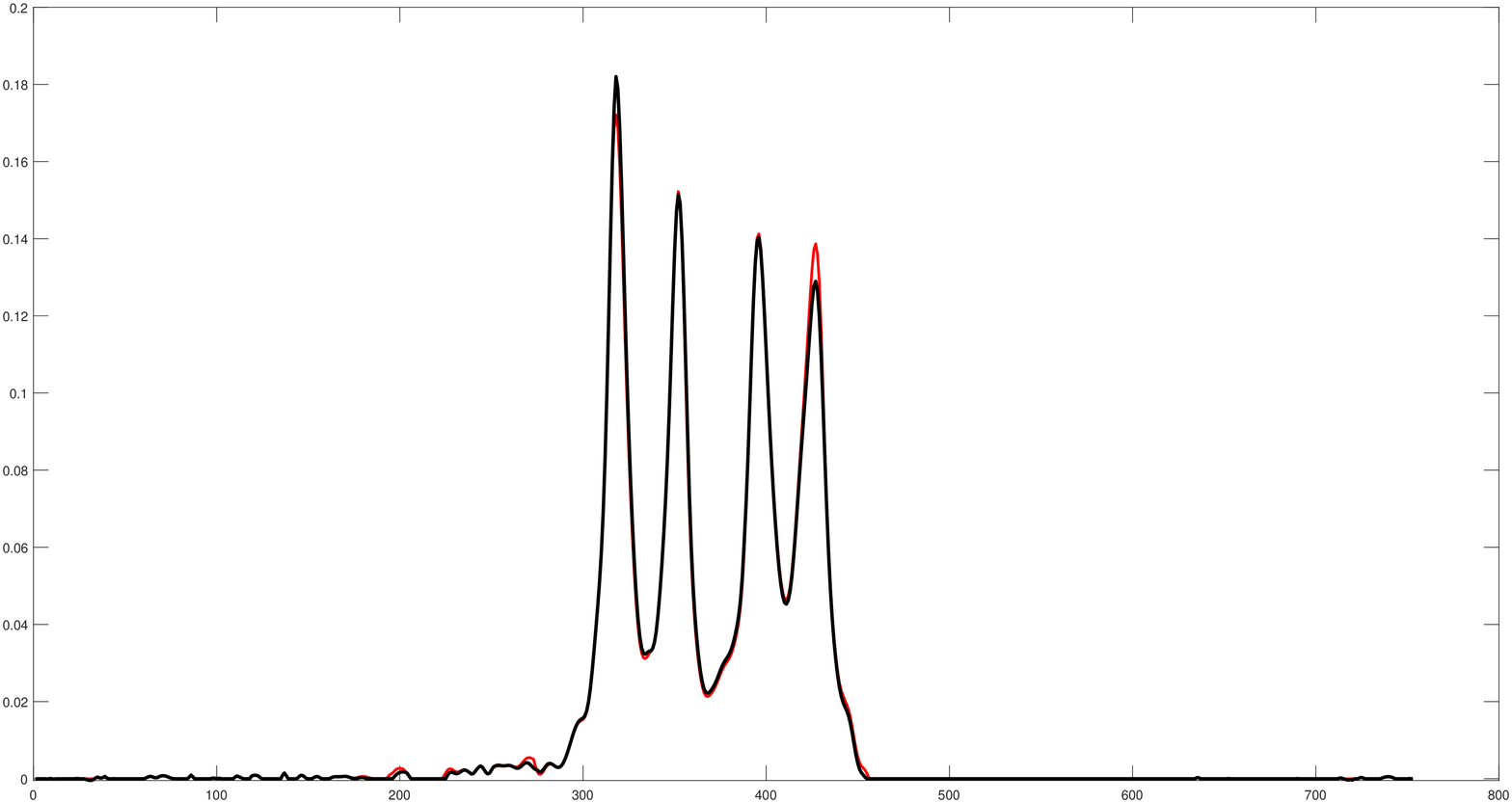}
\includegraphics[height=4cm,width=7cm]{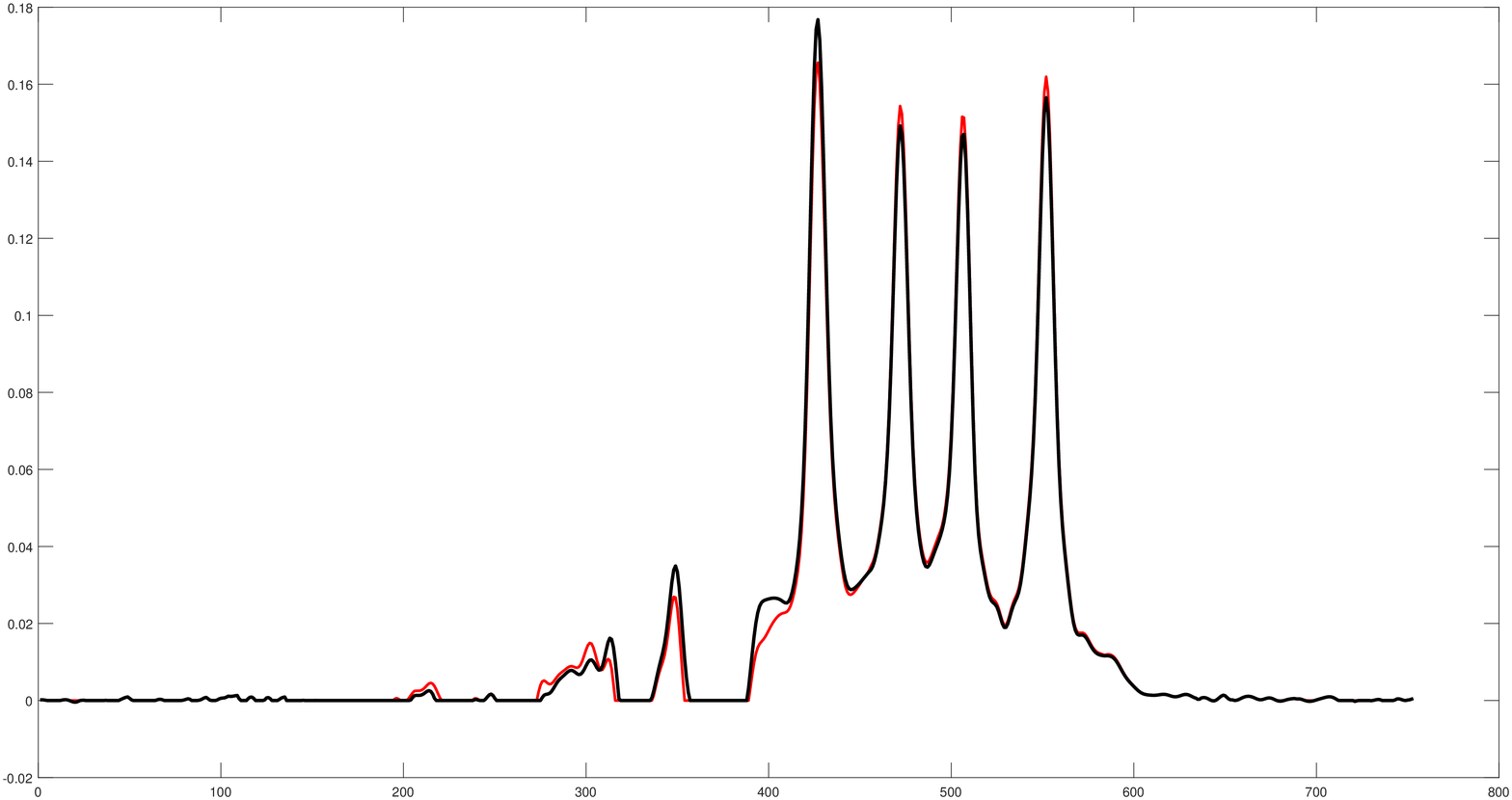}
\\
\includegraphics[height=4cm,width=14cm]{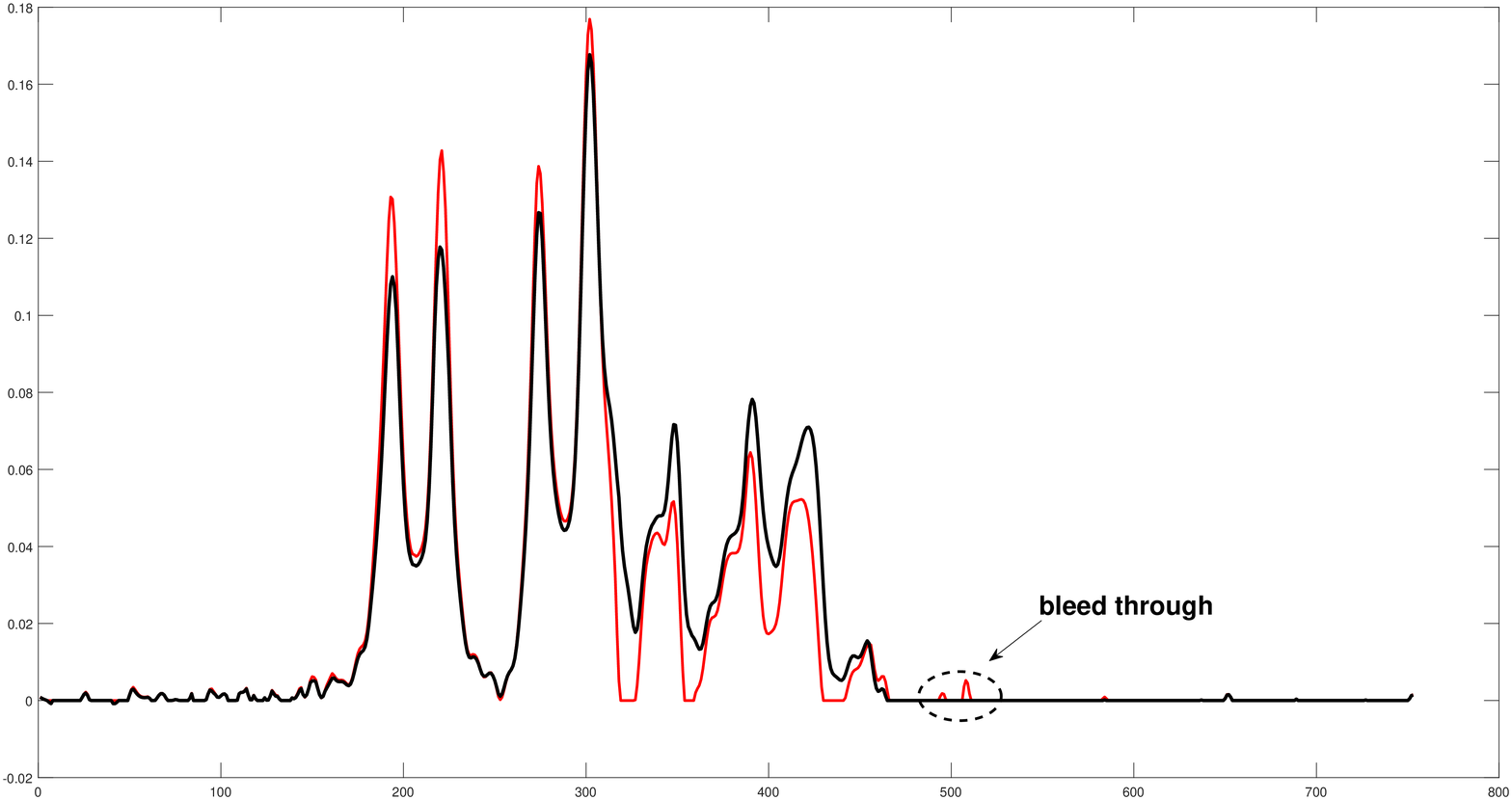}
\caption{Source signals recovered by NN method ({\color{red} red}) and NNP method ({\bf{black}}). Two bleed through peaks in the third signal by NN method can be seen while they are absent in NNP result.}
\label{RecNNvsNNP}
\end{figure}

\section{Conclusion}
This paper presented a preprocessing technique for sparse blind source separation of positive and overlapping data.  Arising in NMR spectroscopy,  the blind source separation problem attempts to unmix the spectral data into a set of basic components (source signals) under a local sparseness condition (the stand alone peaks, or SAP).  Based on the data geometry and SAP, vertex component analysis (such as NN's method) proves to be successful in identifying the mixing matrix whose columns are the edges of the convex cone enclosing the data points.  However, the results of VCA deviate from the truth due to the violation of the SAP in realistic data.  To overcome this problem and improve the separation results, we preprocess the mixture data by a weighted sharpening technique, which manages to enhance the peak resolution by subtracting a constant multiple of its second order derivative.  The fact that the sharpened peaks greatly reduce the violation of SAP source condition lead to an improvement on the identification of the convex cone.   Once an estimate of the mixing matrix is retrieved, the recovery of the source signals can be obtained by a nonnegative least squares (with sparsity constraint if needed).  Besides, we investigate how to tune in the weight parameter and provide an upper bound for this parameter to guide the implementation of the method.  Numerical results on NMR spectra data show satisfactory performance of the proposed method.  For a future line of inquiry, we plan to test and evaluate the method on realistic data from NMR and other spectroscopies, in collaboration with chemists and researchers and based on feedback, further improve the performance and robustness of the algorithms towards real-world applications.
\section*{Acknowledgements}  The authors wish to thank Professor A.J. Shaka and his group for their experimental NMR data.  YS was partially supported by Simons Foundation Grant 800006981.  JX was partially supported by NSF grant IIS-1632935.


\begin{thebibliography}{9999999}
\bibitem{Chang_07} C-I Chang, ed.,
{ ``Hyperspectral Data Exploitation: Theory and Applications"}, Wiley-Interscience, 2007.
\bibitem{choi} S. Choi, A. Cichocki, H. Park, and S. Lee, {\it Blind
source separation and independent component analysis: A review},
Neural Inform. Process. Lett. Rev., 6 (2005), pp. 1--57.
\bibitem{Cic} A. Cichocki and S. Amari, { ``Adaptive Blind Signal and
Image Processing: Learning Algorithms and Applications"}, John Wiley
and Sons, New York, 2005.
\bibitem{Comon} P. Comon,  {\em Independent component analysis--a new concept?}, Signal Processing,
36 (1994) pp. 287--314.
\bibitem{Comon1}P. Comon and C. Jutten, {\em Handbook of Blind
Source Separation: Independent Component
Analysis and Applications}, Academic Press, 2010.
\bibitem{MVT} M. Craig, {\it Minimum-volume transformation for remotely sensed data}, IEEE Transcations on Geoscience and Remote Sensing, 32 (1994), pp.  542--552.
\bibitem{Dul} J.H. Dul\`{a} and R.V. Helgason, {\it A new procedure
for identifying the frame of the convex hull of a finite collection
of points in multidimensional space}, European J. Oper. Res., 92
(1996), pp. 352--367.
\bibitem{nmr} R. Ernst, G. Bodenhausen, and A. Wokaun, {``Principles of Nuclear Magnetic Resonance in One and Two Dimensions"}, Oxford University Press, 1987.
\bibitem{G_O_09} Z. Guo and S. Osher, {\it Template matching via $\ell_1$ minimization and its application to hyperspectral target detection}, Inverse Problems and Imaging, 5 (2011), pp. 19--35.
\bibitem{Kov} L. Kovasznay and H. Joseph, {\it Image processing}, Proc. IRE, 43 (1955),
pp. 560--570.

\bibitem{Lee} D. D. Lee and H. S. Seung,
{\it Learning of the parts of objects by non-negative matrix
factorization}, Nature, 401 (1999), pp. 788--791.

\bibitem{NMF_OR} H. Li, T, Adali, and W. Wang, {\it Non-negative matrix factorization with orthogonality constraints and its
application to Raman Spectroscopy}, The Journal of VLSI Signal Processing Systems for Signal Image and Video Technology 48, (2007), pp. 83--97.

\bibitem{MVSA} J. Li, J.M. Bioucas-Dias, {\it Minimum volume simplex analysis: a fast algorithm to unmix hyperspectral data }, Geoscience and Remote Sensing Symposium, 3 (2008), pp. III-250--III-253.


\bibitem{Miao} H. Miao, {\it Endmember extraction from highly mixed data using minimum volume constrainsted nonnegative matrix factorization}, IEEE Trans. Geosci. Remote Sens., vol. 45(3), 2007, pp. 765--777.


\bibitem{NN05} W. Naanaa and J.--M. Nuzillard, {\it Blind source separation of positive and partially correlated data}, Signal
Processing, 85 (9) (2005), pp. 1711--1722.
\bibitem{NMF0} P. Paatero and U. Tapper, {\it Positive matrix factorization: A non-negative factor model with optimal utilization of error estimates of data
values}, Environmetr., vol. 5, no. 2, 1994, pp. 111--126.
\bibitem{Sch} R. Schachtner, G. P\"{o}pprl, and E. Lang,
{\it Towards unique solutions of non-negative matrix factorization problems by a determinant criterion},
Digit. Signal Process., vol, 21, 2011, pp. 528--534.

\bibitem{sun_xin_pNN} Y. Sun, C. Ridge, F. del Rio, A.J. Shaka, and J. Xin, {\it Postprocessing and sparse blind source separation of positive and partially overlapped data}, Signal Processing, 91(8)(2011), pp. 1838--1851.
\bibitem{YO} W. Yin, S. Osher, D. Goldfarb, J. Darbon,
{\it Bregman iterative algorithm for $\ell_1$-minimization with
applications to compressive sensing}, SIAM J. Image Sci., 1(143) (2008), pp.
143-168.
\end{thebibliography}
\end{document}